\documentclass[runningheads]{llncs}

\usepackage{todonotes}
\usepackage{xcolor}
\usepackage[ruled, vlined, linesnumbered]{algorithm2e}
\usepackage{graphicx}
\usepackage{amsmath, amssymb, xcolor}

\definecolor{blue}{RGB}{170, 170, 255}
\definecolor{lightblue}{RGB}{255,170,170}
\definecolor{orange}{RGB}{170, 255, 170}
\definecolor{green}{RGB}{255, 255, 120}
\newcommand{\Even}{\mathrm{Even}}
\newcommand{\Odd}{\mathrm{Odd}}

\newcommand{\Gc}{\mathcal{G}}

\newcommand{\Tc}{\mathcal{T}}

\newcommand{\Attr}[3]{\mathrm{Attr}_{#1}^{#2} \left( #3 \right)}

\newcommand{\G}{\mathcal G}

\renewcommand{\succ}{\mathrm{succ}}

\renewcommand{\root}{\mathrm{root}}

\newcommand{\graph}[3]{\mathrm{G}_{#1}^{(#2)}( #3 )}

\newcommand{\lazy}[1]{\mathcal{L}\left(#1 \right)}

\newcommand{\rk}[2]{\text{rank}_{#1}(#2)}

\DeclareMathOperator{\before}{before}
\DeclareMathOperator{\after}{after}
\DeclareMathOperator{\level}{level}

\DeclareMathOperator{\scope}{Scope}
\DeclareMathOperator{\subtree}{Subtree}
\DeclareMathOperator{\subtreem}{Subtree^{--}}

\newcommand{\Aset}{A}
\newcommand{\Bset}{B}
\newcommand{\Gset}{G}

\newcommand{\dest}[2]{d_{#1}(#2)}

\newcommand{\edge}[1]{E\left( #1 \right)}

\newcommand{\lazi}[1]{\mathcal{L} \left( #1 \right)}

\newcommand{\out}{\mathrm{out}}
\newcommand{\inbis}{\mathrm{in}}
\newcommand{\rec}{\mathrm{rec}}

\newcommand{\indexx}{\mathrm{index}}

\renewcommand{\P}{P}
\newcommand{\Q}{Q}
\newcommand{\R}{R}

\newcommand{\ES}{\normalfont{\texttt{EmptyScope}}}

\newcommand{\invim}[2]{\left( #1 \right)^{-1}\left( #2 \right)}
\newcommand{\invimnp}[2]{#1^{-1}\left( #2 \right)}

\begin{document}
\title{A symmetric attractor-decomposition lifting algorithm for parity games\thanks{This work has been supported by the EPSRC grant EP/P020992/1 (SolvingParity Games in Theory and Practice).}}
%
%
\author{Marcin Jurdzi\'nski\inst{1} \and
R\'emi Morvan\inst{2} \and
Pierre Ohlmann\inst{3} \and K.~S. Thejaswini\inst{1}}
\authorrunning{Jurdzi\'nski et al.}
%
\institute{Department of Computer Science, 	University of Warwick 
\\ \and ENS Paris-Saclay \\ \and 
IRIF, Universit\'e de Paris
}
\maketitle              
\begin{abstract}

Progress-measure lifting algorithms for solving parity games have the best worst-case asymptotic runtime, but are limited by their asymmetric nature, and known from the work of Czerwi\'nski et al. (2018) to be subject to a matching quasi-polynomial lower bound inherited from the combinatorics of universal trees.

Parys (2019) has developed an ingenious quasi-polynomial McNaughton-Zielonka-style algorithm, and Lehtinen et al. (2019) have improved its worst-case runtime. Jurdzi\'nski and Morvan (2020) have recently brought forward a generic attractor-based algorithm, formalizing a second class of quasi-polynomial solutions to solving parity games, which have runtime quadratic in the size of universal trees.

First, we adapt the framework of iterative lifting algorithms to computing attractor-based strategies. Second, we design a symmetric lifting algorithm in this setting, in which two lifting iterations, one for each player, accelerate each other in a recursive fashion. The symmetric algorithm performs at least as well as progress-measure liftings in the worst-case, whilst bypassing their inherent asymmetric limitation. Thirdly, we argue that the behaviour of the generic attractor-based algorithm of Jurdzinski and Morvan (2020) can be reproduced by a specific deceleration of our symmetric lifting algorithm, in which some of the information collected by the algorithm is repeatedly discarded.  This yields a novel interpretation of McNaughton-Zielonka-style algorithms as progress-measure lifting iterations (with deliberate set-backs), further strengthening the ties between all known quasi-polynomial algorithms to date.

\end{abstract}

\subsubsection*{Context}

Parity games are two-player games on graphs, which have been studied since early 1990's~\cite{EJ91,EJS93} and have many applications in automata theory on infinite trees~\cite{GTW01}, fixpoint logics~\cite{BKMP19,HS19}, verification and synthesis~\cite{LMS19}. 
They are intimately linked to the problems of emptiness and complementation of non-deterministic automata on trees~\cite{EJ91,Zie98}, model checking~\cite{EJS93,HSC16,BW18} and satisfiability checking of fixpoint logics, or fair simulation relations~\cite{EWS05}.

Determining the winner of a parity games are one of the few problems known to lie in the complexity class NP $\cap$ co-NP but not known to have a polynomial algorithm. Existence of a polynomial algorithm for solving parity games, which has been an important open problem for nearly three decades~\cite{EJS93}, has recently gained a lot of attention after the major breakthrough of Calude, Jain, Khoussainov, Li and Stephan~\cite{CJKLS17}, who were the first to provide a quasi-polynomial solution.

\subsubsection*{State-of-the-art.}

Quasi-polynomial solutions to parity games can be classified into two broad family of algorithms: progress measure lifting algorithms~\cite{JL17,FJKSSW19,DJT20} and attractor-based algorithms~\cite{Par19,LSW19,JM20}.

Progress measure lifting algorithms for solving parity games were initially introduced by Jurdzi\'nski~\cite{Jur00-}, and a first quasi-polynomial improvement was discovered by Jurdzi\'nski and Lazi\'c~\cite{JL17}, shortly after the quasi-polynomial breakthrough by Calude et al.~\cite{CJKLS17}. Czerwi\'nski, Daviaud, Fijalkow, Jurdzi\'nski, Lazi\'c, and
Pa\-rys~\cite{CDFJLP19} devoloped the combinatorial notion of universal trees; precisely identifying a structure sufficient for running progress measure lifting algorithms. They showed that
almost all quasi-polynomial algorithms to date, specifically those of Calude et al.~\cite{CJKLS17}, Jurdzi\'nski and Lazi\'c~\cite{JL17} and by Lehtinen~\cite{Leh18}, can be described in this framework, with different constructions of universal trees. Moreover, Czerwi\'nski et al.~\cite{CDFJLP19} also provided a lower bound which matches the succinct construction of a universal tree inherent in the work of Jurdzi\'nski and Lazi\'c~\cite{JL17} and highlighted by Fijalkow~\cite{Fij18}. It is worth noting that although progress measure lifting algorithms achieve the best worst-case complexity, they only focus on constructing a strategy for one of the players, hence suffer from the limitation of being oblivious to the other player's point of view: many parity games (including very trivial examples) induce worst case complexity. This asymmetry makes progress measure lifting algorithms less attractive than attractor-based algorithms for practical applications.

Attractor-based algorithms date back to the classic McNaughton-Zielonka algorithm \cite{McN93,Zie98}, which consistently outperforms other algorithms in practice while having exponential runtime in the worst-case~\cite{Fri11r}. The recent celebrated work of Parys~\cite{Par19} adapts the McNaughton-Zielonka algorithm to reduce its worst-case complexity to quasi-polynomial, and Lehtinen, Shewe and Wojtczak~\cite{LSW19} show how to further refine this technique, lowering the complexity to the square of the state-of-the-art bounds achieved by progress measure algorithms. Although this family of algorithms is not formally subject to the lower bound from universal trees, it is remarked in \cite{LSW19} that a universal tree underlies the structure of their algorithm. Jurdzi\'nski and Morvan~\cite{JM20} have recently proposed a unifying framework for attractor-based algorithms, in the form of a generic algorithm, whose special instances---obtained by using specific families of universal trees---coincide with variants of all the other algorithms. Their algorithm, the \emph{universal attractor algorithm} is a building block for this work, and a pseudo-code is given in the Appendix for convenience (Algorithm~\ref{alg:uad}).

It is well known~\cite{DJL18,JPZ08,McN93,Zie98} that one may extract an attractor-based strategy from an execution of the McNaughton-Zielonka algorithm. In fact, we believe that the overwhelming efficiency of the McNaughton-Zielonka algorithm~\cite{vDij18} might be explained by many parity games having succinct attractor-based strategies, which is one of the motivations for this work. However, the improved quasi-polynomial versions of Parys~\cite{Par19}, Lehtinen et al.,~\cite{LSW19} and Jurdzi\'nski and Morvan~\cite{JM20} do not build strategies, which appears to be an intrinsic limitation to this approach: there is no strong (local) guarantee on subsets output by recursive calls.

\subsubsection*{Contributions and outline.}

Section~\ref{sec:asymmetric} extends the framework of progress measure lifting algorithms to produce attractor-based strategies, generalizing attractor decompositions from~\cite{JM20}. This yields a generic asymmetric attractor decomposition lifting algorithm, Algorithm~\ref{alg:ad-lifting}, which computes the smallest attractor-based strategy in a given tree $\Tc$, in state-of-the-art runtime proportional to the size of $\Tc$. We also provide a construction of a linear graph, in the vocabulary introduced by Colcombet and Fijalkow~\cite{CF19,FGO20}, which formally specifies asymmetric attractor decomposition lifting as a form of generic progress measure lifting.
    
In Section~\ref{sec:symmetric}, we introduce a novel technique which consists of a specific way of running a lifting algorithm for each player in parallel, where each may accelerate the other. This provides a symmetric algorithm, Algorithm~\ref{alg:sym-main}, which matches state-of-the-art complexity in the worst case, while by-passing the limitations of asymmetric iterative algorithms.

Finally, in Section~\ref{sec:attractor-based}, we specify how to decelerate Algorithm~\ref{alg:sym-main} with repetitive loss of information in order to precisely capture the universal attractor decomposition. This gives a first complete formalization of attractor-based algorithms as progress measure lifting iterations, decisively unifying all quasi-polynomial techniques for solving parity games, while suggesting superiority of Algorithm~\ref{alg:sym-main}. As a by-product of the novel techniques that we develop for reasoning about attractor-based algorithms, we obtain an alternative---and more constructive---proof of the main result in \cite{JM20}: the dominion separation theorem, which is a generalization of the key property underlying Parys's breakthrough result.

\section{Preliminaries}\label{sec:prelim}

\subsubsection*{Parity games.}

We provide standard definitions related to parity games.

\begin{definition}
A \emph{parity game} $\Gc$ consists of a finite directed graph $(V, E)$ with no sink, a partition $(V_{\Even}, V_{\Odd})$ of the set of vertices~$V$, and a function $\pi : V \to \{1, \dots, d\}$ that labels every vertex~$v \in V$ with a positive integer~$\pi(v)$ called its \emph{priority}. The size $|\Gc|$ of $\Gc$ is its number of vertices.
\end{definition}

We now define positional strategies, which are also sometimes called memoryless strategies.

\begin{definition}
A \emph{positional Even strategy} is a set $\sigma \subseteq E$
of edges such that:
\begin{itemize}
\item for every $v \in V_{\Even}$, there is an edge $(v, u) \in \sigma$,
\item for every $v \in V_{\Odd}$, if $(v, u) \in E$ then $(v, u) \in \sigma$.
\end{itemize}
Positional Odd strategies are defined by inverting the roles of Even and Odd.
\end{definition}

We may now define what it means to positonally win a parity game from a given vertex.

\begin{definition}
We say that $v$ is \emph{positionaly winning for Even} (resp. for Odd) in $\G$ if there is a positional Even (resp. Odd) strategy $\sigma$, such that any infinite path from $v$ in the subgraph $(V,\sigma)$ has an even (resp. odd) number as maximal priority it visits infinitely often. In this case, we say that $\sigma$ wins from $v$.
\end{definition}

The positional determinacy theorem states that positional strategies are enough to characterise winning sets of a parity game.

\begin{theorem}[Positional determinacy~\cite{EJ91}]
\label{thm:positional-determinacy}
  Every vertex of a parity game $\Gc$ is positionaly winning for one of the players. Moreover, there exists a maximal Even (resp. Odd) positional strategy $\sigma$ which wins from all vertices which are positionaly winning for Even (resp. Odd). 
\end{theorem}

As justified by the previous Theorem, we will now simply write ``$v$ is winning for Even'', rather than ``$v$ is positionaly winning for Even''.

\subsubsection*{Reachability strategies and attractors.}

In a parity game~$\Gc$, for a target set of vertices~$B$ and a set of vertices~$A$ such that $B \subseteq A$, an Even strategy~$\sigma$ is an \emph{Even reachability strategy to $B$ from~$A$} if every infinite path in the subgraph $(V, \sigma)$ that starts from a vertex in~$A$ contains at least one vertex in~$B$. The following definition is standard.

\begin{definition}
Let $B \subseteq V$. We call \emph{Even attractor to $B$ in $\Gc$}, and denote $\Attr{\Even}{\Gc}{B}$, the largest set (with respect to inclusion) from which there is an Even reachability strategy to $B$ in $\Gc$.
\end{definition}

We will need to introduce the refined concept of attracting to a target set $B \subseteq V$ through a safe set $C \supseteq B$. We say that an Even reachability strategy $\sigma$ to $B$ from $A \subseteq C$ stays in $C$ if all infinite path in the subgraph $(V, \sigma)$ remain in $C$ (at least) until reaching a vertex in $B$.

\begin{definition}
Let $B \subseteq C \subseteq V$. We call \emph{Even-attractor to $B$ through $C$} the largest subset of $C$ from which there is an Even reachability strategy to $B$ in $\Gc$ which stays in $C$.
\end{definition}

Odd reachability strategies, and Odd attractors are defined symmetrically.

\subsubsection*{Even and Odd trees and universal trees.}

Throughout the paper, $d$ will refer to a fixed even integer which is an upper bound on the priority of vertices in the parity $\Gc$. We now define Even and Odd trees, which are rooted ordered trees of height $d/2$, which we equip for convenience with Even and Odd levels, respectively. To avoid confusion with vertices from $\Gc$, we will always use the terminology ``node'' to refer to elements of trees.

\begin{definition}
An \emph{Even tree} $\Tc^\Even$ of height $d/2$ is a directed acyclic graph equiped with a map $\level: \Tc^\Even \to \{0, 2, \dots, d\}$ such that
\begin{itemize}
\item each node but one has in-degree 1, the unique node of in-degree 0 is called the \emph{root} of $\Tc^\Even$, it has level $d$ and we denote it by $\root^\Even$,
\item the nodes of out-degree $0$ are called \emph{leaves} and have level 0, and nodes of positive out-degree are called inner nodes,
\item if $(n,n')$ is an arc in $\Tc^\Even$, then $n'$ is said to be a \emph{child} of $n$, and it has level $\level(n) -2$; each inner node $n$ is equiped with a linear order $\leq_n$ over its children.
\end{itemize}
A node $n' \neq n$ which is reachable from $n$ is called a \emph{descendant} of $n$. We define a linear order\footnote{This corresponds to the standard depth-first order in an ordered tree.} $\leq$ over $\Tc^\Even$ which is the unique linear order that satisfies $n_1 < n'_1 < n_2$, for all triplets of nodes such that $n'_1$ is a descendant of $n_1$, and $n_1 \leq_n n_2$ are two. An Odd tree $\Tc^\Odd$ is defined similarly, but with Odd levels $\{1, \dots, d+1\}$. In particular, $\root^\Odd$ has level $d+1$ and all leaves of $\Tc^\Odd$ have level $1$.
\end{definition}

An Even tree is depicted in Figure~\ref{fig:tree}. We now move towards introducing universal trees, for which we need the concept of tree inclusions.

\begin{definition}
Let $\Tc_1$ and $\Tc_2$ be two Even or Odd trees of height $d/2$, and $\phi: \Tc_1 \to \Tc_2$. We say that $\phi$ is an \emph{inclusion} of $\Tc_1$ into $\Tc_2$ if $\phi$ is injective, and $\phi$ respects the ordered-tree structure, that is,  if $n \in \Tc_1$ is an inner node and $n' \leq_n n''$ are two children of $n$, then $\phi(n')$ and $\phi(n'')$ are children of $\phi(n)$ such that $\phi(n') \leq_{\phi(n)} \phi(n'')$.

Inclusions of Odd trees are defined analogously.
\end{definition}

\begin{remark}
Note that a tree inclusion always preserves levels.
\end{remark}

This allows us to define universal trees, following~\cite{CDFJLP19}.

\begin{definition}
Let $\Tc$ be an Even or an Odd tree of height $d/2$, and $n$ an integer parameter. We say that $\Tc$ is $n$-universal if all trees with $\leq n$ leaves are included in $\Tc$.
\end{definition}

\section{Asymmetric attractor decomposition lifting algorithms}\label{sec:asymmetric}

In this section, we adapt the formalism of progress measure lifting algorithm to structures capable of computing attractors-based strategies, such as those output by the McNaughton-Zielonka algorithm and described in~\cite{JM20} as attractor decompositions. To this end, we shall introduce so-called \emph{embeded} attractor decompositions. These are similar to progress measures in that they are labellings of $V$ by positions in a structured set which satisfy local validity conditions. However, the structure that we use, a precise enrichment of trees with additional positions, inspired by the so-called lazifications of~\cite{DJL19} and~\cite{DJT20}, is tailored to capture attractor-based strategies.

Subsection~\ref{subsec:labellings}, introduces the structure that allows to define embedded attractor decompositions. Subsection~\ref{subsec:lifting} then adapts progress measure lifting algorithms to the attractor-based setting, by introducing attractor decomposition lifting algorithms. Finally, Subsection~\ref{subsec:pm}, explains that the vocabulary of universal graphs from~\cite{CF19} allows to formally capture attractor decomposition liftings as generic progress measure liftings as in~\cite{FGO20}. This is independant from the remainder of the paper (Sections~\ref{sec:symmetric} and~\ref{sec:attractor-based}).

\subsection{Labellings and embedded attractor decompositions}\label{subsec:labellings}

\subsubsection*{Lazification of an Even or Odd tree.}

Let $\Tc$ be an Even or an Odd tree of height $d/2$. We now introduce an ordered set $\lazi{\Tc} \supseteq \Tc$ which enriches the structure of the tree $\Tc$ by adding additional positions between nodes. This will allow us to define attractor-based strategies as labellings of $V$ by elements of $\lazi{\Tc}$.

To avoid confusion, we use the word ``positions'' to refer to elements of $\lazi{\Tc}$. In particular, a node $n \in \Tc$ is also a position in $\lazi{\Tc}$. Figure~\ref{fig:tree} illustrates an Even tree $\Tc^\Even$, and its lazification $\lazi{\Tc^\Even}$. 

Let $\Tc$ be an Even or Odd tree. We define $\lazi{\Tc} \supseteq \Tc$ by adding to $\Tc$ new positions, $\before(n)$ and $\after(n)$ for all each node $n \in \Tc$, and identifying $\after(n'')$ and $\before(n')$ for all pairs $(n',n'')$ of consecutive children of some inner node $n$.

We define levels for positions in $\lazi{\Tc}$ by extending $\level$ over $\Tc$, and having $\level(\after(n))= \level(\before(n))=\level(n)+1$. The elements of $\lazi{\Tc}$ are linearly ordered by the unique linear order $\leq$ which extends the order over $\Tc$ and satisfies, for all nodes $n \in \Tc$ with children $n', \before(n) < n < \before(n') < n' < \after(n') < \after(n)$.

\begin{figure}[h]
  \makebox[\textwidth][c]{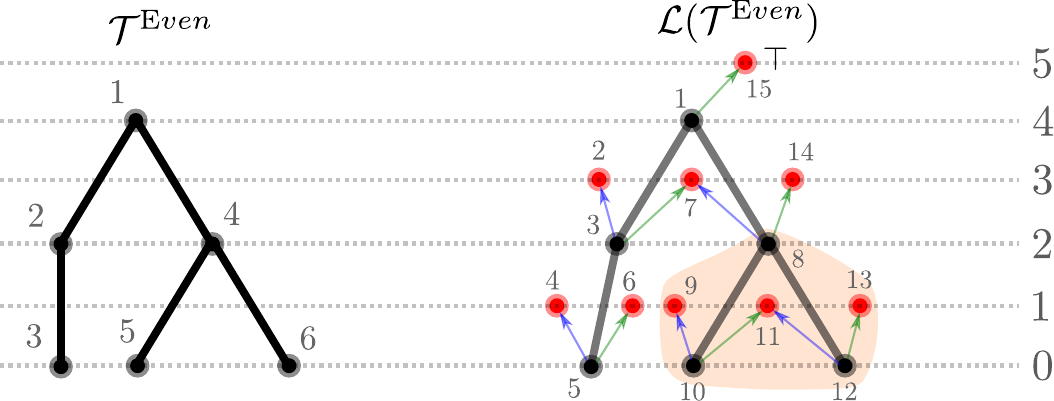}%
\caption{An Even tree $\Tc^\Even$ of height $2$ ($d=4$), and its lazified version $\lazi{\Tc^\Even}$. The local orders over children of inner nodes in $\Tc^\Even$ are depicted left-to-right, and the obtained linear orders on $\Tc^\Even$ and $\lazi{\Tc^\Even}$ are indicated using integers. Regular positions in $\lazi{\Tc^\Even}$ are depicted in black, and lazy positions are the red ones. The blue and green arrows respectively represent the $\before$ and $\after$ maps. Finally, $\subtree(n)$, where $n$ is the 8-th position in $\lazi{\Tc^\Even}$ corresponds to the orange zone. \label{fig:tree}}
\end{figure}

We use $\top$ to refer to $\after(\root)$, and moreover, we supress the position $\before(\root)$ for convenience.

We will refer to positions of $\lazi{\Tc} \setminus \Tc$ as \emph{lazy}, and positions of $\Tc \subseteq \lazi{\Tc}$ as \emph{regular}. If $n \in \Tc$ is a regular position, we use $\subtree(n)$ to denote the subset of $\lazi{\Tc}$ which consists of all positions between $n$ (included) and $\after(n)$ (excluded), and $\subtreem(n)$ to denote $\subtree(n) \setminus \{n\}$.

\subsubsection*{Labellings and validity.}

We are now ready to define Even and Odd labellings, which are the fundamental we work with.

\begin{definition}\label{def:labelling}
We say that a map $\mu$ from $V$ to $\lazy{\Tc}$ is a \emph{labelling} if it satisfies the following: 
\begin{itemize}
    \item for all regular positions $n \in \lazy{\Tc}$ all vertices mapped to $n$ have priority $\level(n)$, 
    \item for all lazy positions $\ell \in \lazy{\Tc} \setminus \Tc$, all vertices mapped to $\ell$ have priority $\leq \level(\ell)$,
\end{itemize}
If $\Tc$ is and Even tree (resp. an Odd tree), we will refer to such $\mu$ as an Even (resp. Odd) labelling.
\end{definition}

\begin{remark}
Note that vertices mapped to $\subtree(n)$ must have priority $\leq \level(n)$, and vertices mapped to $\subtreem(n)$ must have priority $< \level(n)$.
\end{remark}

We now introduce some useful notations. Given an position $p \in \lazy{\Tc}$, we call $\succ(p)$ its immediate successor with respect to the order over $\lazy{\Tc}$. Note that if $n \in \Tc$ is an inner node and $n_1$ is its smallest child, it holds that $\succ(n)=\before(n_1)$. If $\mu$ is a labelling and $p \in \lazi{\Tc}$, we also use the notation $\mu^{-1}(\leq p)$ (resp. $\mu^{-1}(< p)$) to refer to the set of vertices which $\mu$ maps to positions $\leq p$ (resp. $< p$). 

We turn to the concepts of validity in a labelling $\mu$. These are a collection of local conditions that together ensure that $\mu$ describes a winning attractor-based strategy over all vertices that are not mapped to $\top$.

We start by defining validity of an edge in a given Even or Odd labelling.

\begin{definition}
Let $e =(u,v) \in E$ be an edge in $\Gc$ , and $\mu : V \to \lazi{\Tc}$ be an Even (resp. Odd) labelling. We say that $e$ is valid in $\mu$ if
\begin{itemize}
\item $\mu(u)$ is regular and $\mu(v) < \after(\mu(u))$,
\item $\mu(u)=p$ is lazy and $e$ is part an Even (resp. Odd) reachability strategy to $\mu^{-1}(< p)$ which stays in $\mu^{-1}(\leq p)$, or
\item $\mu(u) = \top$.
\end{itemize}\label{def:validEdge}
\end{definition}

\begin{remark}
The second condition is equivalent to ``$\mu(v) < \mu(u)$ or $\mu(v)=\mu(u)$ and Even (resp. Odd) can ensure to reach $\mu^{-1}(< \mu(u))$ from $v$ while remaining in $\mu^{-1}(\leq \mu(u))$''. 
\end{remark}

We now define validity of a vertex.

\begin{definition}
Let $v \in V$ and $\mu : V \to \lazi{\Tc}$ be an Even (resp. Odd) labelling. We say that $v$ is valid in $\mu$ if either 
\begin{itemize}
    \item $v$ is an Even (resp. Odd) vertex and it has an outgoing valid edge, or
    \item $v$ is an Odd (resp. Even) vertex and all of its outgoing edges are valid.
\end{itemize}
\end{definition}

Given an Even (resp. Odd) labelling $\mu$, the set $\sigma_\mu \subseteq E$ of valid edges in $\mu$ defines a positional Even (resp. Odd) strategy in $\Gc$ if all vertices are valid in $\mu$.

\begin{lemma}\label{lem:ad-is-win}
Let $\mu : V \to \lazi{\Tc}$ be an Even (resp. Odd) labelling, and assume that all vertices are valid in $\mu$. Then $\sigma_\mu$ wins from all vertices which are not mapped to $\top$.
\end{lemma}

\begin{proof}[Proof sketch]
One first proves the following statement: if $n \in \Tc$ is a node with children $n' \leq n''$, then any path in $\sigma_\mu$ from $\mu^{-1}(\subtree(n'))$ to $\mu^{-1}(\subtree(n''))$ must visit a vertex of Even priority $\geq \level(n)$. Let $v \in V$ be such that $\mu(v) \neq \top$ and consider an infinite path $\pi \in V^\omega$ from $v$. Clearly $\pi$ never visits $\top$, since no valid edge can lead to $\top$. We then let $n$ be the node of lowest level such that $\pi$ remains in $\mu^{-1}(\subtree(n))$. Then the previous statement implies that $\pi$ visits $\mu^{-1}(n)$ infinitely often, and otherwise only vertices of priority $\leq \level(n)$, so the maximal priority seen infinitely often is $\level(n)$ which is even.
\end{proof}

\begin{definition}[Embedded attractor decomposition]
An Even (resp. Odd) embedded attractor decomposition in an Even (resp. Odd) tree $\Tc$ is an Even (resp. Odd) labelling $\mu: V \to \lazi{\Tc}$ in which all vertices are valid.
\end{definition}

\begin{remark}\label{rmk:ad-def}
This departs from the original definition of attractor decomposition proposed in~\cite{JM20}: in our embedded version, we allow lazy positions to only \emph{include} specific attractors slighly generalizing the original definition requires equality. 
\end{remark}

Lemma~\ref{lem:ad-is-win} states that attractor decompositions imply winning strategies. The converse also holds, and is stated in the following Lemma. It is easily obtained by adapting the known results over progress measures~\cite{EJ91,Jur00-} to this setting, essentially restricting to regular positions, or alternatively by constructing an attractor decomposition from an execution of the McNaughton-Zielonka algorithm, as in~\cite{JM20}.

\begin{lemma}\label{lem:existence-of-ad}
There exists an attractor decomposition embedded in an Even (resp. Odd) tree with at most $|\Gc|$ leaves, which maps all vertices which are winning for Even (resp. Odd) to a position $< \top$.
\end{lemma}

\subsection{Principles of attractor decomposition lifting}\label{subsec:lifting}

Just like progress measure lifting algorithms, attractor decomposition lifting algorithms are asymmetric procedures which perform successive updates to either an Even or an Odd labelling until all vertices are valid, effectively producing an embedded attractor decomposition. In this section, we give more details about this process, prove its correctness, and discuss its complexity.

The following Proposition, which also holds in the context of progress measures, allows for efficient proofs of the two main ingredients for a lifting algorithm, namely a closure by point-wise minimum (Lemma~\ref{lem:min}), and monotonicity of the destination (Lemma~\ref{lem:monotonicity-of-destination}).

\begin{proposition}\label{prop:small-technical}
Let $\mu: V \to \lazi{\Tc}$ and $\mu': V \to \lazi{\Tc}$ be two labellings. We assume that $\mu \geq \mu'$ (pointwise, with respect to the order on $\lazi{\Tc}$), that $u \in V$ is such that $\mu(u) = \mu'(u)$, and that $u$ is valid in $\mu$. Then $u$ is valid in $\mu'$.
\end{proposition}

\begin{proof}
Let $e=(u,v) \in \edge{\Gc}$ be an edge in $\Gc$ which is valid in $\mu$. We prove that $e$ is valid in $\mu'$, which implies the wanted result.
\begin{itemize}
\item If $\mu(u)$ is a regular node, then $\mu'(v) < \mu(u) < \after(\mu(u)) = \after(\mu(v))$, so $e$ is valid in~$\mu'$.
\item If $\mu(u)=\ell$ is a lazy node, since $\invimnp{\mu}{< \ell} \subseteq \invim{\mu'}{< \ell}$ and $\invimnp{\mu}{\leq \ell} \subseteq \invim{\mu'}{\leq \ell}$, an attracting edge to $\invimnp{\mu}{< \ell}$ through $\invimnp{\mu}{\leq \ell}$ is also an attracting edge to $\invim{\mu'}{< \ell}$ through $\invim{\mu'}{< \ell}$.
\end{itemize}
\end{proof}

With Proposition~\ref{prop:small-technical} in hands, we are ready to prove closure by minimality.

\begin{lemma}\label{lem:min}
Let $\mu_1: V \to \lazi{\Tc}$ and $\mu_2: V \to \lazi{\Tc}$ be two attractor decompositions embedded in $\Tc$. Let $\mu : V \to \lazi{\Tc}$ be the pointwise minimum of $\mu_1$ and $\mu_2$, with respect to the order on $\lazi{\Tc}$. Then $\mu$ is an attractor decomposition embedded in $\Tc$.
\end{lemma}

\begin{proof}
It is clear that $\mu$ indeed defines a labelling. We now prove its validity. Let $v \in V$, and let $i \in \{1,2\}$ be such that $\mu(v)=\mu_i(v)$. Then $\mu_i \geq \mu$, and $v$ is valid in $\mu_i$, so we conclude that $v$ is valid in $\mu$ by applying Proposition~\ref{prop:small-technical}.
\end{proof}

Note that the labelling which maps all vertices to $\top$ is valid. Hence, by Lemma~\ref{lem:min}, the minimal attractor decomposition embedded in a given tree $\Tc$ is well defined.

\begin{remark}
It is not hard to see that in a minimal attractor decompositions $\mu$ embedded in $\Tc$, if $\ell \in \lazi{\Tc}$ is a lazy position, then $\mu^{-1}(\ell)$ is precisely the attractor to $\mu^{-1}(< \ell)$ through $\mu^{-1}(\leq \ell)$. In that way, we recover (see Remark~\ref{rmk:ad-def}) the original definition from~\cite{JM20} as far as minimal embedded attractor decompositions are concerned.
\end{remark}

We now define the very important concept of destination of a vertex in a labelling.

\begin{definition}
Let $\mu$ be a labelling, and $v \in V$. We call \emph{destination} of $v$ in $\mu$ the smallest position $r \geq \mu(v)$ such that $v$ is valid in the labelling $\mu'$ which satisfies $\mu'(v) = r$ and is everywhere else identical to $\mu$. We denote it by $\dest{\mu}{v}.$
\end{definition}

Note that $\dest{\mu}{v} \geq \mu(v)$, with equality if and only if $v$ is valid in $\mu$. We have the following monotonicity property over destinations.

\begin{lemma}\label{lem:monotonicity-of-destination}
Let $\mu_1,\mu_2$ be labellings in $\Tc$, such that $\mu_1 \leq \mu_2$, and let $v \in V$. Then it holds that $\dest{\mu_1}{v} \leq \dest{\mu_2}{v}$.
\end{lemma}

\begin{proof}
It suffices to prove that $v$ is valid in the labelling $\mu'$ which satisfies $\mu'(v)=\dest{\mu_2}{v}$ and is everywhere else identical to $\mu_1$. Let $\mu'_2$ be the labelling such that $\mu'_2(v)=\dest{\mu_2}{v}$, and $\mu'_2$ is identical to $\mu_2$ elsewhere. By definition of $\dest{\mu_2}{v}$, $v$ is valid in $\mu'_2$. Now it holds that $\mu'_2 \geq \mu'$, so Proposition~\ref{prop:small-technical} concludes.
\end{proof}

Given a labelling $\mu$ and a vertex $v$ which is invalid in $\mu$, we call ``lifting $v$ in $\mu$'' the operation of augmenting $\mu(v)$ to a larger value, which is $\leq \dest{\mu}{v}$. The general asymetric lifting algorithm, Algorithm~\ref{alg:ad-lifting}, is very simply described as follows: perform any arbitrarily chosen lift, and iterate, until every vertex is valid.

\begin{remark}
Let $\Tc^\Even$ be an Even tree of height $d/2$ and $n$ be its root. Then the smallest labelling into $\Tc^\Even$ is easy to characterize: vertices of priority $\level(\root)=d$ are mapped to $\root \in \lazi{\Tc}$, while vertices of priority $<d$ are mapped to $\succ(\root) \in \lazi{\Tc}$. In an Odd tree of height $d/2$, the smallest labelling maps all vertices to $\succ(\root)$.
\end{remark}

\begin{algorithm}[htb]
  \DontPrintSemicolon
Let $\mu \leftarrow$ smallest labelling of $V(\Gc)$ into $\Tc$ \label{lin:init} \;
\While{there is an invalid vertex $v$ in $\mu$}{
Lift $v$ in $\mu$ \;
}
\Return $\mu$ \;
\caption{The asymmetric attractor decomposition lifting algorithm in an Even or Odd tree $\Tc$.}
\label{alg:ad-lifting}
\end{algorithm}

The following Lemma, analogous to the well-known property of progress-measure lifting algorithms, states correctness of Algorithm~\ref{alg:ad-lifting}.

\begin{lemma}\label{lem:ad-lifting-correct}
Algorithm~\ref{alg:ad-lifting} outputs the smallest attractor decomposition of $\G$ embedded in $\Tc$.
\end{lemma}

\begin{proof}
Let $\lambda: V \to \lazi{\Tc}$ be the smallest attractor decomposition of $\Gc$ embedded in $\Tc$. We prove by induction on the execution of the algorithm, that its output $\mu$ satisfies $\mu \leq \lambda$. This holds trivially at initialisation on line~\ref{lin:init}, and is maintenained by performing lifts a lift at vertex $v$. Indeed, since $v$ is valid in $\lambda$, and by monotonicity of destinations, we have $\dest{\mu}{v} \leq \dest{\lambda}{v} = \lambda(v)$. Upon exiting the while loop, it holds that $\mu$ is valid, so by minimality of $\lambda$, $\mu \geq \lambda$. Hence, $\mu = \lambda$.
\end{proof}

In particular, any vertex which is not mapped to $\top$ in the obtained embedded attractor decomposition $\lambda$ is winning for Even. The following theorem states that if the tree $\Tc$ is large enough then the converse also holds.

\begin{theorem}
If $\Tc^\Even$ is $|\G|$-universal, then Algorithm~\ref{alg:ad-lifting} returns an attractor decomposition $\lambda$ such that $\lambda^{-1}(< \top)$ is the Even winning set in $\Gc$.
\end{theorem}

\begin{proof}[Proof sketch]
By Lemma~\ref{lem:existence-of-ad}, there is an embedded attractor decomposition $\mu_0$ in an Even tree $\Tc_0^\Even$ with at most $|\G|$ leaves, such that $\mu_0^{-1}(\leq \top_0)$ is the Even winning set in $\Gc$. By composition with the inclusion of $\Tc_0^\Even$ in $\Tc^\Even$ (inclusions over trees extend into inclusions over their lazified versions), we obtain an attractor decomposition $\mu$ embedded in $\Tc^\Even$ such that $\mu^{-1}(\leq \top_0)$ is the Even winning set. We conclude by minimality of $\lambda$.
\end{proof}

We now comment on the complexity of Algorithm~\ref{alg:ad-lifting}. The following is easy to see since each lift updates a vertex to a stricly greater position in $\lazi{\Tc}$

\begin{lemma}
An execution of Algorithm~\ref{alg:ad-lifting} performs at most $|\G||\lazi{\Tc}| = O(|\G||\Tc|)$ lifts in total.
\end{lemma}

In succint constructions of universal trees, operations such as checking validity or performing lifts can be done in polylogarithmic time~\cite{JL17,DJT20}, hence incur no significant complexity overhead to the above bounds. In particular, asymmetric attractor decomposition lifting iterations achieve state of the art worst-case runtime.

\subsection{Formal similarities with progress measure lifting algorithms}\label{subsec:pm}

Using the vocabulary developped in~\cite{CF19} and~\cite{FGO20} enables us to formally describe attractor decomposition lifting algorithms as progress measure lifting algorithms inside a specific novel construction of linear (universal) graphs. As this contribution is independant from the remainder of the paper, we present it in Appendix~\ref{app:progress-measure}.

\section{A symmetric lifting algorithm}\label{sec:symmetric}

We will now introduce a novel symmetric lifting algorithm. Roughly, it runs in parallel two asymmetric attractor decomposition lifting algorithms, one for each player, in a specific and constrained manner. However, it goes beyond the scope of Algorithm~\ref{alg:ad-lifting} by making the two labellings interact to accelerate the process. We first define the important concepts that are needed to describe and visualize the algorithm, and explain the ideas behind it (Subsection~\ref{subsec:intro-to-alg}). We then provide a pseudo-code, and formally prove correctness in Subsection~\ref{subsec:pseudo-code-correctness}. Finally, Subsection~\ref{subsec:complexity} gives a complexity upper bound.

\subsection{Introducing the algorithm}\label{subsec:intro-to-alg}

Let $d$ be an even integer, and fix a parity game $\G$ with priorities in $\{1, \dots, d\}$. We will also fix an Even tree $\Tc^\Even$ of height $d/2$, and an Odd tree $\Tc^\Odd$ of height $d/2$. Recall that $\Tc^\Even$ has levels in $\{0,2,\dots, d\}$ while $\Tc^\Odd$ has levels in $\{1, 3, \cdots, d+1\}$. We use $\top^\Even$ (resp. $\top^\Odd$) and $\root^\Even$ (resp. $\root^\Odd$) to refer to the largest and smallest vertices of $\Tc^\Even$ (resp. $\Tc^\Odd$).

The universal attractor decomposition algorithm from~\cite{JM20} explores the trees $\Tc^\Even$ and $\Tc^\Odd$ in a depth-first manner, by alternating between each trees. The next definition formalizes the structure of its recursive calls in a product structure.

\begin{definition}
We define the interleaving $\Tc$ of $\Tc^\Even$ and $\Tc^\Odd$ as a rooted levelled ordered tree with levels in $\{1, \dots, d\}$ as follows:
\begin{itemize}
\item The nodes in $\Tc$ are given by all pairs $(n^\P, n^\Q) \in \Tc^\Even \times \Tc^\Odd \cup \Tc^\Odd \times \Tc^\Even$ such that $\level_\P(n^\P) = \level_\Q(n^\Q)-1$.
\item A node $n=(n^\P, n^\Q)$ of $\Tc$ has level $\level(n)=\level(n^\P)$.
\item A node $n=(n^\P, n^\Q)$ of $\Tc$ has children $(n^\Q_1,n^\P), \dots, (n^\Q_k,n^\P),$ in this order, where $n^\Q_1 \leq \dots \leq n^\Q_k$ are the $k$ ordered children of $n^\Q$ in $\Tc^\Q$. If $n^\Q$ is a leaf in $\Tc^\Q$, then $n$ is a leaf in $\Tc$.
\end{itemize}
\end{definition}

The algorithm stores an Even labelling $\mu^\Even : V \to \lazi{\Tc^\Even}$ inside $\Tc^\Even$, and an Odd labelling $\mu^\Odd: V \to \lazi{\Tc^\Odd}$ inside $\Tc^\Odd$. We will use the notation $\mu=(\mu^\Even, \mu^\Odd)$ to refer to such a pair of labellings. For a vertex $v \in V$, we have $\mu(v)=(\mu^\Even(v), \mu^\Odd(v)) \in \lazi{\Tc^\Even} \times \lazi{\Tc^\Odd}$.

\subsubsection*{Visualizing the execution on the grid}

It is very convenient to visualize the evolution of the pair of labellings $\mu$ on a 2-dimensional grid, where the axes represent $\lazi{\Tc^\Even}$ and $\lazi{\Tc^\Odd}$, respecting the linear orders we have equiped them with. On this grid, and at any given point in the algorithm, one may display for each vertex $v \in V$, the point with coordinates $\mu(v)=(\mu^\Even(v), \mu^\Odd(v))$. Lifts in $\mu^\Even$ and $\mu^\Odd$ then respectively correspond to positive vertical or horizontal shifts on the grid. We equip the grid $\lazi{\Tc^\Even} \times \lazi{\Tc^\Odd}$ with the (pointwise) partial order $\leq$ defined by $r_1 = (r_1^\Even, r_1^\Odd) \leq r_2 = (r_2^\Even, r_2^\Odd)$ if and only if $r_1^\Even \leq r^2_\Even$ and $r_1^\Odd \leq r_2^\Odd$.

\begin{definition}
Let $\mu=(\mu^\Even, \mu^\Odd)$ be a pair of labellings and $v \in V$ be a vertex. 
\begin{itemize}
\item We define the destination $\dest{\mu}{v}$ of $v$ in $\mu$ to by $\dest{\mu}{v} =(\dest{\mu^\Even}{v}, \dest{\mu^\Odd}{v}) \in \lazi{\Tc^\Even} \times \lazi{\Tc^\Odd}$.
\item We call ``lifting $v$ in $\mu$'' the operation of updating $\mu(v)$ to a stricly greater value which is~$\leq \dest{\mu}{v}$.
\end{itemize}
\end{definition}

\begin{remark}
Note that a lift can be performed at $v$ in $\mu$ if and only if $d \neq \mu(v)$, that is, $v$ is invalid in either $\mu^\Even$ or $\mu^\Odd.$
\end{remark}

Figure~\ref{fig:grid-lifts} depicts the grid $\lazi{\Tc^\Even} \times \lazi{\Tc^\Odd}$, and illustrates lifting at $v$ in $\mu$.

\begin{figure}[h]
  \makebox[\textwidth][c]{\scalebox{.8}{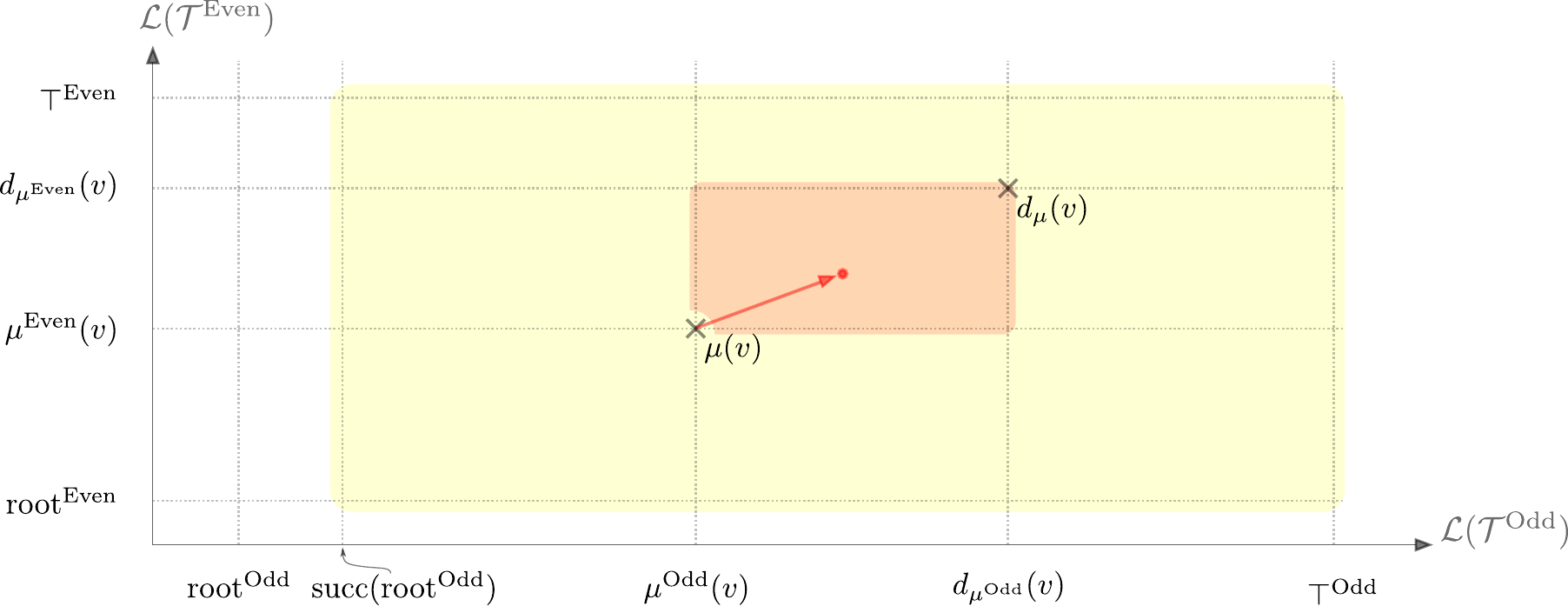}}
\caption{Visualizing the evolution of the algorithm on the grid $\lazi{\Tc^\Even} \times \lazi{\Tc^\Odd}$. The yellow box represents possible locations for $\mu(v)$. Note that it excludes the absciss $(\root^\Odd)$, since $\Gc$ has no vertex of priority $d+1.$ The crosses depict $\mu(v)$ and its destination $\dest{\mu}{v}$ in $\mu$ (which depends on the location of other vertices), and the red box represents possible arrivals for a lift at $v$. 
\label{fig:grid-lifts}}
\end{figure}

Consider a vertex $v \in \Gc$. Since $v$ cannot be winning for both players, we know that regardless of $\Tc^\Even$ and $\Tc^\Odd$, it holds that for (at least) one player $\R \in \{\Even, \Odd\}$, the smallest $\R$ attractor decomposition $\lambda^\R$ embedded in $\Tc^\R$ maps $v$ to $\top^\R$. Roughly, our algorithm will perform successive lifts in both labellings, until all vertices satisfy either $\mu^\Even(v)=\top^\Even$ or $\mu^\Odd(v)=\top^\Odd$. In other words, for each vertex $v$, $\mu(v)$ undergoes successive shifts until it hits either the right-most or the top-most border.

To make for a clear description of the algorithm, we need to define some subsets of $\lazi{\Tc^\Even} \times \lazi{\Tc^\Odd}$. These are represented on the grid in Figure~\ref{fig:alg-sym}.

\begin{definition}
Let $n=(n^\P, n^\Q) \in \Tc$ be a node of the interleaving $\Tc$ of $\Tc^\Even$ and $\Tc^\Odd$. We let $\times_\P$ denote the regular product if $\P=\Even$, and $(X,Y) \mapsto Y \times X$ otherwise. We define four subsets of $\lazi{\Tc^\Even} \times \lazi{\Tc^\Odd}$. The two latter are only defined if $n$ is not the root of $\Tc$.
\begin{itemize}
\item $\scope(n) = \subtree(n^\P) \times_\P \subtreem(n^\Q)$,
\item $\Gset(n) = \subtree(n^P) \times_P \{n^\Q\} $
\item $\Bset(n) =  \{\before(n^\P)\} \times_\P \subtree(n^\Q)$, and
\item $\Aset(n) = \{\after(n^\P)\} \times_\P \subtreem(n^\Q)$.
\end{itemize}
Note that if $n \leq n'$ are two successive children of the same node in $\Tc$, then $\Bset(n') = \Aset(n)$.
\end{definition}

\begin{remark}
By Definition~\ref{def:labelling}, a pair of labellings $\mu=(\mu^\Even, \mu^\Odd)$, can only map vertices of priority $\leq \level(n)$ to $\scope(n)$, and vertices of priority $\leq \level(n)+1$ to $\Bset(n)$ or $\Aset(n)$. Note moreover that no vertex can be mapped to $\Gset(n)$, since this would require having priority $= \level(n^Q)$ and $\leq \level(n^\P)=\level(n^\Q) -1$.
\end{remark}

Given a subset $S \subseteq \lazi{\Tc^\Even} \times \lazi{\Tc^\Odd}$, we will informally say that one has ``emptied $S$'' if the considered pair of labellings $\mu$ maps no vertex to $S$. With this terminology, we may now describe the execution of the algorithm. 

\subsubsection*{Execution of the algorithm}

We present our algorithm recursively, the main recursive procedure being \ES$(n, \mu)$, which updates in place the pair of labellings $\mu=(\mu^\Even, \mu^\Odd)$ over vertices in $\mu^{-1}(\scope(n)).$ The purpose of a call to \ES$(n, \mu)$ is to empty $\scope(n)$.

\begin{figure}[h]
  \makebox[\textwidth][c]{\scalebox{.8}{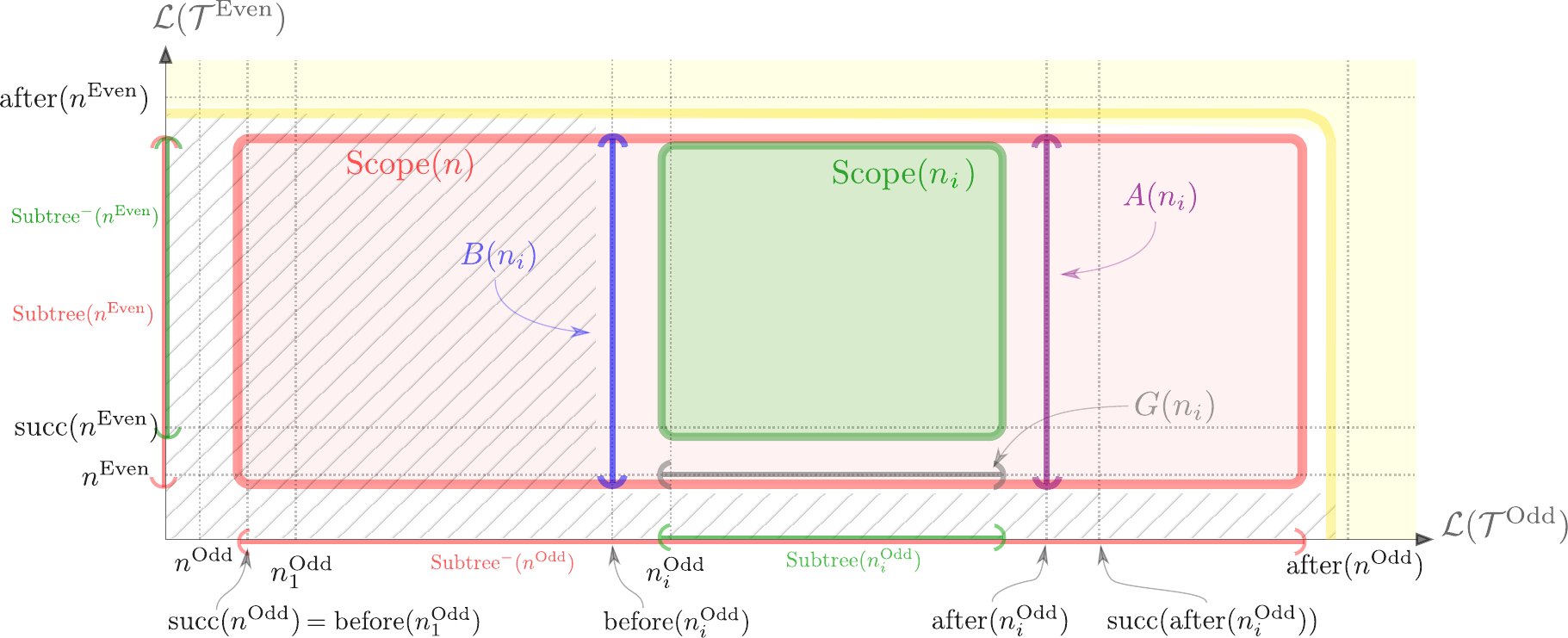}}%
\caption{Illustration of the $i$-th iteration of a call to \ES, in the case where $\P=\Even$. We assume (this is guaranteed over the course of the algorithm) that the hatched subset has no vertex mapped to it. Note that $\succ(\after(n^\Odd_i))$ is equal to $n^\Odd_{i+1}$ if $i<k$, and to $\after(n^\Odd)$ otherwise.
\label{fig:alg-sym}}
\end{figure}

Note that $\scope(n)$ is partitionned into 
\[
\scope(n)=\bigsqcup_{i=1}^k (\Gset(n_i) \sqcup \Bset(n_i) \sqcup \scope(n_i)) \sqcup \Aset(n_k),
\]
where $n_1, \dots, n_k$ are the children of $n$.

The algorithm will iteratively empty each of these subsets. As we have remarked, a labelling cannot map any vertex to $\Gset(n_i)$, so we shall successively empty $\Bset(n_i)$ and $\scope(n_i)$ for $i \in \{1, \dots, k\}$, and then finally $\Aset(n_k)$.

The key insight to understanding the algorithm is that emptying $\Bset(n_i)$ for $i \in \{1, \dots, k \}$ and $\Aset(n_k)$ can be performed efficientely, that is, by performing exactly one lift at each vertex that is mapped in the set. This is non-trivial and argued for in Subsection~\ref{subsec:pseudo-code-correctness}. Emptying $\scope(n_i)$ is handled by performing a recursive call.

In the algorithm we have described so far, there is no interaction between the two labellings. In other words, the procedure above can be stated as a very specific lifting strategy in the asymetric setting for either player. We now present what we consider as the strength of this symmetrical approach, namely acceleration.

This works as follows: at the start of a call to \ES$(n, \mu)$, if for either player $\R \in \{\Even, \Odd\}$, it holds that all vertices that lie in $\scope(n)$ are valid for $\R$, then we may immediately update the labelling of the opponent $\bar R$ to set all vertices from $\scope(n)$ to $\after(n^{\bar \R})$. Indeed, in a (non-accelerating) execution of \ES$(n, \mu)$, lifts would all be parallel to the $\bar \R$-axis by validity in $\R$, so the only way of emptying $\scope(n)$ is by updating (only) $\mu^{\bar \R}$ to positions $\geq \after(n^\R)$.

\subsection{Pseudo-code and proof of correctness}\label{subsec:pseudo-code-correctness}

We are now ready to expose some pseudo-code for Algorithm~\ref{alg:sym-main}.

\begin{algorithm}[htb]
\DontPrintSemicolon
  \SetKwFunction{emptyScope}{$\text{EmptyScope}$}
  \SetKwProg{fun}{procedure}{:}{}

  \fun{\emptyScope{$n, \mu$}}{
        $n \leftarrow (n^\P, n^\Q)$ \;
  		\If{all vertices in $\mu^{-1}(\scope(n))$ are valid in $\mu^\R$ for some $\R$\label{lin:ifacc}}{
			\ForAll{$v \in \mu^{-1}(\scope(n))$}{
				$\mu^{\bar \R}(v) \leftarrow \after(n^{\bar \R})$ \; 
			}
  		}
		\Else{\label{lin:else}
     		Let $n_1, \dots, n_k \leftarrow$ children of $n$ in $\Tc$ \;
  			\For{$i \leftarrow 1$ \KwTo $k$}{
				\While{there is $v \in \mu^{-1}(\Bset(n_i))$ such that $\dest{\mu}{v} \notin \Bset(n_i)$}{ \label{lin:while}
					Lift $v$ in $\mu$ to a position $\notin \Bset(n_i)$ \label{lin:ub}\;
				}
				\emptyScope$(n_i, \mu)$ \label{lin:recurse}\;
				}
			\While{there is $v \in \mu^{-1}(\Aset(n_k))$ such that $\dest{\mu}{v} \notin \Aset(n_k)$}{\label{lin:sndwhile}
				Lift $v$ in $\mu$ to a position $\notin \Aset(n_k)$ \label{lin:ua}\;
			}
  		}
  			
  	}
  
  \tcc{Main procedure:}
  Let $\mu^\Even,\mu^\Odd \leftarrow $ smallest labellings in $\Tc^\Even, \Tc^\Odd$ respectively \;
  Let $n \leftarrow $ root of $\Tc$ \;
  \emptyScope{$n, (\mu^\Even, \mu^\Odd)$}\;
  \Return{$(\mu^\Even, \mu^\Odd)$} \;
  \caption{The symmetric algorithm\label{alg:sym-main}}
\end{algorithm}

\subsubsection*{Correctness of the algorithm}

Theorem~\ref{thm:correctness} states that Algorithm 2 returns a pair  of labellings in which all vertices are mapped to one of the two $\top$ positions, while being smaller that the respective minimal embedded attractor-decompositions. A careful reader will observe that these two guarentees can be thought of as an strengthened version of the dominion separation theorem in~\cite{JM20}.

\begin{theorem}\label{thm:correctness}
Algorithm~\ref{alg:sym-main} returns a pair of labellings $\mu=(\mu^\Even, \mu^\Odd)$ such that
\begin{itemize}
\item for all vertices $v \in V$, either $\mu^\Even(v) = \top^\Even$ or $\mu^\Odd(v)=\top^\Odd$, and
\item for both $\R \in \{\Even, \Odd\}$, we have $\mu^\R \leq \lambda^\R$, where $\lambda^\R$ denotes the smallest $\R$-attractor decomposition of $\Gc$ embedded in $\Tc^\R$.
\end{itemize}
\end{theorem}

A detailed proof, which in spirit amounts to showing that $\Bset(n_i)$ and $\Aset(n_k)$ are indeed emptied respectively in lines~\ref{lin:while} and~\ref{lin:sndwhile}, is given in Appendix~\ref{app:correctness}.

\subsection{Complexity analysis}\label{subsec:complexity}

We now provide complexity upper bounds for Algorithm~\ref{alg:sym-main}. We implicitly assume that computing destinations and performing lifts is performed efficiently, hence only concentrate on bounding the total number of recursive calls.

\begin{theorem}
Algorithm~\ref{alg:sym-main} performs at most $O(|\G|d \min(|\Tc^\Even|, |\Tc^\Odd|))$ calls to \ES.
\end{theorem}

\begin{proof}[Sketch]
We say that a call to \ES$(\mu,n)$ is accelerating if it holds that all vertices in $\mu^{-1}(\scope(n))$ are valid for one of the player. We make the two following observations: first, there may be atmost $|\G|d$ successive accelerating calls. Second, if a call is non-accelerating, then either a lift is performed \emph{in both labellings} in the $|\Gc|$ following recursive calls, or, one of the following $|\Gc|$ recursive calls is non-accelerating and at a smaller level. These two facts together imply the wanted upper bound. See Appendix~\ref{app:complexity} for a full proof.
\end{proof}

\section{Relations to attractor-based algorithms}\label{sec:attractor-based}

We now introduce a variant of Algorithm~\ref{alg:sym-main}, which exactly simulates the universal attractor decomposition algorithm. Subsection~\ref{subsec:resets} introduces the variant, and Subsection~\ref{subsec:simul} formally proves equivalence with known algorithms.

\subsection{Performing short-lifts and resets}\label{subsec:resets}

We will now present a determinisic variant of Algorithm~\ref{alg:sym-main} which is tailored to simulate the universal attractor decomposition algorithm from~\cite{JM20}. First, we explain how to perform so called ``short-lifts'', to resolve the non-determinism of Algorithm~\ref{alg:sym-main}. Then, we introduce resets, which are key to rigorously capture attractor computations. We describe the variant informally, a complete pseudo-code is given in Appendix~\ref{app:equivalence}.

Note that the order in the choice of $v$ in the while-loops, as well as the precise arrival of the lifts, can be chosen arbitrarily in Algorithm~\ref{alg:sym-main}. What we now describe as ``performing short-lifts'' is a way of resolving this non-determinism. In both while-loops (lines~\ref{lin:while} and~\ref{lin:sndwhile}), we prioritise vertices $v$ such that $\dest{\mu^\P}{v} \geq \after(n^\P)$, which we lift in $\mu^\P$ to $\after(n^\P)$. Once remaining vertices $v$ satisfy $\dest{\mu^\P}{v} < \after(n^P)$, we proceed to lift them in $\mu^\Q$. In the while-loop of line~\ref{lin:sndwhile}, remaining vertices are now lifted in $\mu^\Q$ to $\after(n^\Q)$. In the while-loop of line~\ref{lin:while}, we first lift those that have destination $\geq \after(n^\Q)$ are to $\after(n^\Q)$, and then the remaining vertices are lifted in $\mu^\Q$ to the smallest possible position: $n_i^\Q$ for those of priority $\level(n_i^\Q)$, and $\succ(n_i^\Q)$ for those of smaller priority.

In contrast to short-lifts, resets are not prescribed by asymmetric attractor decomposition lifting algorithms as defined in Section~\ref{sec:asymmetric}. This operation, which is artificial in the context of lifting algorithms, reverts the labelling of some well chosen vertices to smaller positions. Formally, let $\mu_\out$ be the pair of labellings obtained at the very end of a call to \ES$(n,\mu_\inbis)$, after the while loop of line~\ref{lin:sndwhile}. Then when adding resets to Algorithm~\ref{alg:sym-main}, we revert the value of $\mu_\out^\Q$ to $\succ(n^\Q)$ for all vertices mapped to $\{\after(n^\P)\} \times_\P \subtreem(n^\Q)$. Figure~\ref{fig:alg-sym-with-resets}, in Appendix~\ref{app:equivalence} illustrates the variant (including resets).

We state correctness of the algorithm with short-lifts and resets.

\begin{lemma}
The algorithm with short-lifts and resets terminates with a pair of labellings $\mu = (\mu^\Even, \mu^\Odd)$ such that
\begin{itemize}
\item for all vertices $v \in V$, either $\mu^\Even(v) = \top^\Even$ or $\mu^\Odd(v) = \top^\Odd$, and
\item for both $\R \in \{\Even, \Odd\}$, we have $\mu^\R \leq \lambda^\R$, where $\lambda^\R$ denotes the smallest $\R$-attractor decomposition of $\Gc$ embedded in $\Tc^\R$.
\end{itemize}
\end{lemma}

\begin{proof}
The second item is clear since performing short-lifts and resets also preserve being smaller that the smallest embedded attractor decompositions for either player. For the first item, the proof of Lemma~\ref{lem:main-inductive}, in Appendix~\ref{app:correctness} can be read also in this context: short-lifts are formally authorized by Algorithm~\ref{alg:sym-main}, and it is straightforward that resets do not break the invariant stated in Lemma~\ref{lem:main-inductive}.
\end{proof}

As we shall argue in the next subsection, the algorithm with short-lift and resets simulates the universal attractor decomposition algorithm~\cite{JM20}. It is worth noting that the main technical and conceptual contribution of~\cite{JM20}, namely, the \emph{dominion separation theorem}, which provides a precise inductive statement for the proof of the universal attractor decomposition theorem, is a corollary of Lemma~\ref{lem:main-inductive}.

\begin{remark}
Note however that we lose the complexity upper bound when performing resets (it still holds however if only short-lifts are performed, since these are formally allowed by the semantics of Algorithm~\ref{alg:sym-main}). When performing short-lifts and resets, the best upper bound we provide on the number of calls to \ES~is given by the size of the interleaving $\Tc$, which is $O(|\Tc^\Even||\Tc^\Odd|)$.
\end{remark}
%
%
%
%
%

\subsection{Simulating attractor-based algorithms}\label{subsec:simul}

We are now ready to argue that the variant of Algorithm~\ref{alg:sym-main} with short-lifts and resets behaves just like the universal attractor decomposition algorithm. Intuitively, when performing short-lifts and reset, the positions of the vertices in the scope are constrained to some very precise positions (see Figure~\ref{fig:alg-sym-with-resets}), and the subsets computed in the universal attractor decomposition correspond to the inverse image by $\mu$ of these positions. See the proof of Lemma~\ref{lem:inductive-eq} in Appendix~\ref{app:equivalence} for full details. We refer to the pseudo-code given in the Appendix for the universal attractor decomposition algorithm (Algorithm~\ref{alg:uad}).

\begin{theorem}\label{thm:eq}
Let $\Gc$ be a parity game with priorities in $\{1, \dots, d\}$, and $\Tc^\Even$ and $\Tc^\Odd$ be an Even tree of height $d$ and an Odd tree of height $d+1$ respectively. Let $\mu_\out$ and $D$ be the respective outputs of Algorithm~\ref{alg:variant} and Algorithm~\ref{alg:uad} when ran on $\Gc, \Tc^\Even$ and $\Tc^\Odd$, and $t,t'$ their respective total number of recursive calls. Then
\begin{itemize}
\item $\mu_\out^\Odd$ maps all vertices of $D$ to $\top^\Odd$ and $\mu_\out^\Even$ maps all vertices of $\Gc \setminus D$ to $\top^\Even,$ and
\item $t \leq t' \leq (\delta+1)t$, where $\delta$ is an upper bound to the degrees of $\Tc^\Even$ and $\Tc^\Odd$.
\end{itemize}
\end{theorem}

\section{Conclusion and perspective}

We have introduced a novel symmetric lifting algorithm which, when diminished with repetitive loss of information, captures the generic attractor-based algorithm of~\cite{JM20}. This formally reconciles the two long-standing families of progress measure lifting and attractor-based algorithms, which together include all known quasi-polynomial algorithms for solving parity games to date. We now comment on the un-diminished version of Algorithm~\ref{alg:sym-main}, which we believe opens exciting opportunities for furthur research.

First, for its practical applicability. Lehtinen and Boker~\cite{Leh18,LB19} have conjectured that relevant classes of parity games encountered in practice have small register-number. Building on this work, Daviaud, Jurdzi\'nski and Thejaswini~\cite{DJT20} have shown that register-number and Strahler number coincide, while exhibiting succinct trees which are Strahler-universal. This motivates a promising program of solving parity games by iteratively running algorithms parametrized with succinct trees which are universal for small parameters. Algorithm~\ref{alg:sym-main} is particularly fit for this use, as it inherits the practical efficiency of attractor-based algorithms, while producing labellings---with separation guarantees---which in turn can be fed to successive calls.

Second, for its theoretical aspects. Though crucial to justifying the pertinence of Algorithm~\ref{alg:sym-main}, the upper-bound in roughly $\min(\Tc^\Even, \Tc^\Odd)$ provided in~\ref{subsec:complexity} somehow falls short of effectively accounting for the recurrent interaction between the two labellings. Indeed, running two lifting algorithms, one for each player, independently from each other, and stopping whenever one of the two iterations terminate, yields the same complexity upper bound. We believe that the symmetric and recursive nature of Algorithm~\ref{alg:sym-main}, eminently inspired by the attractor-based paradigm, might lead to a provable gain in the worst-case complexity.

 \newpage

\subsubsection*{Acknowledgements}

We express our most sincere gratitude to Nathana\"el Fijalkow and Olivier Serre for their careful proofreading and valuable suggestions.

 \bibliographystyle{splncs04}
 \bibliography{bib}

\begin{thebibliography}{10}
\providecommand{\url}[1]{\texttt{#1}}
\providecommand{\urlprefix}{URL }
\providecommand{\doi}[1]{https://doi.org/#1}

\bibitem{BKMP19}
Baldan, P., K\"onig, B., Mika-Michalski, C., Padoan, T.: Fixpoint games on
  continuous lattices. Proceedings of the ACM on Programming Languages
  \textbf{3}(POPL, January 2019),  26:1--26:29 (2019)

\bibitem{BW18}
Bradfield, J.C., Walukiewicz, I.: Handbook of Model Checking, chap. The
  mu-calculus and model checking, pp. 871--919. Springer (2018)

\bibitem{CJKLS17}
Calude, C.S., Jain, S., Khoussainov, B., Li, W., Stephan, F.: Deciding parity
  games in quasipolynomial time. In: STOC 2017. pp. 252--263. ACM, Montreal,
  QC, Canada (2017)

\bibitem{CF18}
Colcombet, T., Fijalkow, N.: Parity games and universal graphs. CoRR
  \textbf{abs/1810.05106} (2018), \url{http://arxiv.org/abs/1810.05106}

\bibitem{CF19}
Colcombet, T., Fijalkow, N.: Universal graphs and good for games automata: New
  tools for infinite duration games. In: Foundations of Software Science and
  Computation Structures - 22nd International Conference, {FOSSACS} 2019.
  Lecture Notes in Computer Science, vol. 11425, pp. 1--26. Springer (2019)

\bibitem{CDFJLP19}
Czerwi\'nski, W., Daviaud, L., Fijalkow, N., Jurdzi\'nski, M., Lazi\'c, R.,
  Parys, P.: Universal trees grow inside separating automata:
  {Q}uasi-polynomial lower bounds for parity games. In: Thirtieth Annual
  ACM-SIAM Symposium on Discrete Algorithms, SODA 2019. pp. 2333--2349. SIAM,
  San Diego, CA (2019)

\bibitem{DJL18}
Daviaud, L., Jurdzi\'nski, M., Lazi\'c, R.: A pseudo-quasi-polynomial algorithm
  for mean-payoff parity games. In: 33rd Annual ACM/IEEE Symposium on Logic in
  Computer Science, LICS 2018. pp. 325--334. ACM, Oxford, UK (2018)

\bibitem{DJL19}
Daviaud, L., Jurdzi\'nski, M., Lehtinen, K.: Alternating weak automata from
  universal trees. In: 30th International Conference on Concurrency Theory,
  CONCUR 2019. Leibniz International Proceedings in Informatics (LIPIcs),
  vol.~140, pp. 18:1--18:14. Schloss Dagstuhl -- Leibniz-Zentrum f\"ur
  Informatik, Amsterdam, the Netherlands (2019)

\bibitem{DJT20}
Daviaud, L., Jurdziński, M., Thejaswini, K.S.: {The Strahler Number of a
  Parity Game}. Leibniz International Proceedings in Informatics (LIPIcs),
  vol.~168, pp. 123:1--123:19 (2020)

\bibitem{vDij18}
van Dijk, T.: Oink: {A}n implementation and evaluation of modern parity game
  solvers. In: Tools and Algorithms for the Construction and Analysis of
  Systems, 24th International Conference, TACAS 2018. LNCS, vol. 10805, pp.
  291--308. Springer, Thessaloniki, Greece (2018)

\bibitem{EJ91}
Emerson, E.A., Jutla, C.S.: Tree automata, mu-calculus and determinacy. In:
  32nd Annual Symposium on Foundations of Computer Science. pp. 368--377. IEEE
  Computer Society, San Juan, Puerto Rico (1991)

\bibitem{EJS93}
Emerson, E.A., Jutla, C.S., Sistla, P.: On model-checking for fragments of
  $\mu$-calculus. In: CAV 1993. LNCS, vol.~697, pp. 385--396. Springer,
  Elounda, Greece (1993)

\bibitem{EWS05}
Etessami, K., Wilke, T., Schuller, R.A.: Fair simulation relations, parity
  games, and state space reduction for {B}\"uchi automata. SIAM Journal on
  Computing  \textbf{34}(5),  1159--1175 (2005)

\bibitem{FJKSSW19}
Fearnley, J., Jain, S., de~Keijzer, B., Schewe, S., Stephan, F., Wojtczak, D.:
  An ordered approach to solving parity games in quasi-polynomial time and
  quasi-linear space. International Journal on Software Tools for Technology
  Transfer  \textbf{21}(3),  325--349 (2019)

\bibitem{Fij18}
Fijalkow, N.: An optimal value iteration algorithm for parity games. CoRR
  \textbf{abs/1801.09618} (2018), \url{http://arxiv.org/abs/1801.09618}

\bibitem{FGO20}
Fijalkow, N., Gawrychowski, P., Ohlmann, P.: Value iteration using universal
  graphs and the complexity of mean payoff games. In: 45th International
  Symposium on Mathematical Foundations of Computer Science, {MFCS} 2020 (2020)

\bibitem{Fri11r}
Friedmann, O.: Recursive algorithm for parity games requires exponential time.
  RAIRO --- Theor. Inf. and Applic.  \textbf{45}(4),  449--457 (2011)

\bibitem{GTW01}
Gr\"adel, E., Thomas, W., Wilke, T. (eds.): Automata, Logics, and Infinite
  Games: A Guide to Current Research, LNCS, vol.~2500. Springer (2002)

\bibitem{HSC16}
Hasuo, I., Shimizu, S., C\^{\i}rstea, C.: Lattice-theoretic progress measures
  and coalgebraic model checking. In: POPL 2016. pp. 718--732. ACM, St.
  Petersburg, FL, USA (2016)

\bibitem{HS19}
Hausmann, D., Schr\"oder, L.: Computing nested fixpoints in quasipolynomial
  time. arXiv:1907.07020 (2019)

\bibitem{JL17}
Jurdzi\'nski, M., Lazi\'c, R.: Succinct progress measures for solving parity
  games. In: 32nd Annual ACM/IEEE Symposium on Logic in Computer Science, LICS
  2017. pp.~1--9. IEEE Computer Society, Reykjavik, Iceland (2017)

\bibitem{JM20}
Jurdzi\'nski, M., Morvan, R.: A universal attractor decomposition algorithm for
  parity games. arXiv:2001.04333 (2020)

\bibitem{JPZ08}
Jurdzi\'nski, M., Paterson, M., Zwick, U.: A deterministic subexponential
  algorithm for solving parity games. SIAM Journal on Computing
  \textbf{38}(4),  1519--1532 (2008)

\bibitem{Jur00-}
Jurdzi{\'n}ski, M.: Small progress measures for solving parity games. In:
  Annual Symposium on Theoretical Aspects of Computer Science. pp. 290--301.
  Springer (2000)

\bibitem{Leh18}
Lehtinen, K.: A modal $\mu$ perspective on solving parity games in
  quasi-polynomial time. In: 33rd Annual ACM/IEEE Symposium on Logic in
  Computer Science, LICS 2018. pp. 639--648. IEEE, Oxford, UK (2018)

\bibitem{LB19}
Lehtinen, K., Boker, U.: Register games. arXiv:1902.10654 (October 2019)

\bibitem{LSW19}
Lehtinen, K., Schewe, S., Wojtczak, D.: Improving the complexity of parys'
  recursive algorithm. CoRR  \textbf{abs/1904.11810} (2019),
  \url{http://arxiv.org/abs/1904.11810}

\bibitem{LMS19}
Luttenberger, M., Meyer, P.J., Sickert, S.: Practical synthesis of reactive
  systems from {LTL} specifications via parity games. arXiv:1903.12576 (2019)

\bibitem{McN93}
McNaughton, R.: Infinite games played on finite graphs. Annals of Pure and
  Applied Logic  \textbf{65}(2),  149--184 (1993)

\bibitem{Par19}
Parys, P.: Parity games: {Z}ielonka's algorithm in quasi-polynomial time. In:
  MFCS 2019. Leibniz International Proceedings in Informatics (LIPIcs),
  vol.~138, pp. 10:1--10:13. Schloss Dagstuhl -- Leibniz-Zentrum f\"ur
  Informatik, Aachen, Germany (2019)

\bibitem{Zie98}
Zielonka, W.: Infinite games on finitely coloured graphs with applications to
  automata on infinite trees. Theoretical Computer Science  \textbf{200}(1--2),
   135--183 (1998)

\end{thebibliography}

\newpage

\appendix

\section*{Appendix}

\section{Attractor decomposition lifting is a progress measure lifting} \label{app:progress-measure}

This appendix provides an equivalent way of constructing the lazification $\lazi{\Tc}$ of a tree $\Tc$, but in the vocabulary of~\cite{CF19} and~\cite{FGO20}. Intuitively, this corresponds to the construction of~\cite{CF18}, with additional components for computing attractors.

This allows to formally describe attractor decomposition lifing algorithms, and by extension, our symmetric algorithm, Algorithm~\ref{alg:sym-main}, in the well known context of progress measure lifting algorithms. Adapting to the vocabulary of universal graphs, we now describe parity games labelling the edges with priorities.

We give a construction, which given an Even tree $\Tc$ of height $d/2$ and an integer $N$, provides a linear graph $\graph{N}{\Even}{\Tc}$, with the following features:
\begin{itemize}
\item if $\Tc$ is an $N$-universal tree, then $\graph{N}{\Even}{\Tc} $ is an $N$-universal graph,
\item the number of vertices of $\graph{N}{\Even}{\Tc}$ is bounded by $2N|\Tc|$,
\item a labelling $\mu: \Gc \to \lazi{\Tc}$ can be converted to a graph homomorphism $\phi: \Gc \to \graph{|\Gc|}{\Even}{\Tc}$, and vice-versa,
\item validity (from Definition~\ref{def:validEdge}) of $\mu$ as and validity of $\phi$ (as a progress measure as in the sense of~\cite{FGO20}) are equivalent.
\end{itemize}

The construction is given by induction on the even integer $d$. If $d=0$, then $\Tc$ is reduced to its root, and we let $\graph{N}{\Even}{\Tc}$ be the graph with a single vertex, and a loop labelled by 0. Assume $d \geq 2$, and assume constructed $\graph{N}{\Even}{\Tc_1}, \graph{N}{\Even}{\Tc_k}$, where $\Tc_1, \dots, \Tc_k$ are the Even trees of height $d-2$ obtained respectively by restrincting $\Tc$ to the  vertices reachable from $n_1, \dots n_k$, the children of the root in this order. We describe $\graph{N}{\Even}{\Tc}$, which is ilustrated in Figure~\ref{fig:graph}, by the following:
\begin{itemize}
\item vertices of $\graph{N}{\Even}{\Tc}$ are the disjoint union of vertices in $\graph{N}{\Even}{\Tc_i}$ for $i \in \{1, \dots, k\}$, together with $(k+1) N$ additional fresh vertices $(0,0), \dots, (0,N-1), (1,0), \dots, (1,N-1), \dots, (k,0), \dots, (k, N-1)$,
\item all pairs of vertices in $\graph{N}{\Even}{\Tc}$ are connected by an edge labelled by $d$, and all edges from $\graph{N}{\Even}{\Tc_i}$, for $i \in \{1, \dots, k\}$ appear in the corresponding copy
\item there are edges labelled by all $h \in \{1, \dots, d-1\}$ from $u$ to $v$ if $\rk{d}{u} > \rk{d}{v}$, where $\rk{d}{\cdot}$ is $2i+1$ for a vertex of the form $(i,l)$, and $2i$ for a vertex in $\graph{N}{\Even}{\Tc_i}$.
\end{itemize}

\begin{figure}[h]
  \makebox[\textwidth][c]{\scalebox{.8}{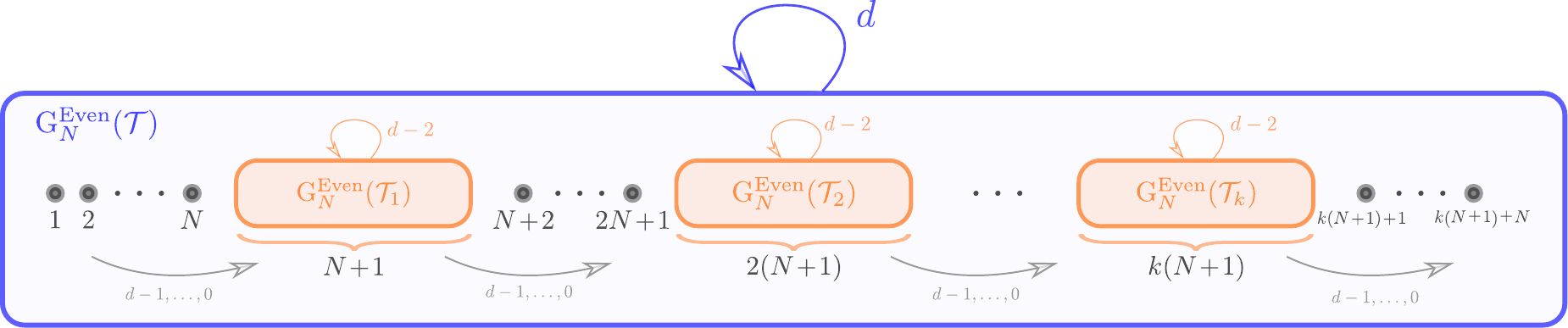}}%
\caption{Depiction of the construction $\graph{N}{\Even}{\Tc}$. The rank is displayed below each vertex. All vertices are connected by an edge of priority $d$ (in blue), and edges of priority $0, \dots, d-1$ follow strict increases in the rank. Remaining edges (among which the orange ones) are defined inductively in each copy $\graph{N}{\Even}{\Tc_i}$. \label{fig:graph}}
\end{figure}

\section{Correctness of Algorithm~\ref{alg:sym-main} (Theorem~\ref{thm:correctness})} \label{app:correctness}

The second item in the Theorem is clear: if it were not for acceleration, this would be exactly the same proof (for each player) as that of~\ref{lem:ad-lifting-correct}. Moreover, it is straightforward that acceleration can only produce a smaller pair of labellings, since vertices are updated (only) in $\mu^{\bar \R}$ to $\after(n^{\bar \R})$ when accelerating, while they would be updated to a position $\geq  \after(n^{\bar \R})$ if we were not accelerating.

Hence we focus on the first item. The following Lemma is our main inductive hypothesis, which formalizes the following: under the right inductive invariant, a call to \ES$(n, \mu)$ empties $\scope(n)$.

\begin{lemma}\label{lem:main-inductive}
Let $\mu_\inbis$ be a pair of labellings, $n \in \Tc$, and let $V_\inbis = \mu_\inbis^{-1}(\scope(n))$. We assume that all vertices $v \notin V_\inbis$ satisfy $\mu_\inbis^\R(v) \geq \after(n^\R)$ for some $\R$. Let $\mu_\out$ be the labelling obtained after an execution of \ES$(n, \mu)$.\emph{ Then it holds that $\mu_\out^{-1}(\scope(n)) = \varnothing$.}
\end{lemma}

Theorem~\ref{thm:correctness} follows by applying Lemma~\ref{lem:main-inductive} to the root of $\Tc$. The following is a small technical lemma which is instrumental in the proof of Lemma~\ref{lem:main-inductive}.

\begin{proposition} \label{prop:one-step-strat}
Let $\R \in \{\Even, \Odd\}$, and $\mu^{\R}$ be a $\R$-labelling in $\Tc^\R$. Let $p \in \lazi{\Tc^\R}$, and $v \in V$. Assume that $\bar{\R}$ has a strategy to reach $(\mu^\R)^{-1}(\geq p)$ in one step\footnote{The terminology $\bar \R$ can reach $T$ in one step from $v$'' is just a more efficient way of stating either $v$ belongs to $\bar \R$ and has an edge towards $T$ or $v$ belongs to $\R$ and all of its outgoing edges lead to $T$''.} from $v$, and that $\mu^\R(v) < p$. Then $v$ is invalid in $\mu^\R$ and has destination $\geq p$.
\end{proposition}

\begin{proof}
Since $\mu^\R(v) < r$, any edge from $v$ to $\invim{\mu^\R}{\geq r}$ is invalid in $\mu^\R$. An easy case disjunction, whether $v$ belongs to $\R$ or not, concludes that $v$ is invalid in $\mu^\R$. This holds whatever the value of $\mu^\R(v)$, provided it is $< r$, so the destination of $v$ in $\mu^\R$ is necessarily $\geq r$.
\end{proof}

We are now ready to prove Lemma~\ref{lem:main-inductive}. We encourage the reader to refer to Figure~\ref{fig:alg-sym}, which depicts the subsets needed in the proof, in the case where $\P=\Even$.

\begin{proof}
If $n$ is a leaf, then the result is trivial: since $\Gc$ has no vertices of priority 0, $V_\inbis= \varnothing$. Let $n=(n^\P, n^\Q) \in \Tc$ be an inner node, and assume the result of the Lemma known for the children $n_1, \dots, n_k$ of $n$. Assuming there is $\R$ such that all vertices of $V_\inbis$ are valid in $\R$, we may immediately conclude: there is an acceleration, and all vertices $v \in V_\inbis$ satisfy $\mu^{\bar \R}(v) = \after(n^{\bar \R}),$ which directly implies the result. Hence we assume otherwise.

Let us introduce more notations. For $i \in \{1, \dots, k\}$, we let $\mu_{\inbis, i}, \mu_{\rec, i}$ and $\mu_{\out, i}$ respectively denote the labelling at the beginning of the $i$-th iteration of the for-loop, just before the recursive call on line~\ref{lin:recurse} and just after the recursive call. Note that $\mu_{\inbis, 1} = \mu_{\inbis}$, and for $i < k$, $\mu_{\inbis, i} = \mu_{\out, i+1}$. For all $\indexx \in \{\inbis, \rec, \out\} \times \{1, \dots, k\}$, we also let $V_\indexx= \mu_{\indexx}^{-1}(\scope(n))$. Note that by definition, it holds that for all vertices $v \in V_{\inbis, 1}$, we have $\mu_{\inbis, 1}^\Q(v) \geq \before(n_1^\Q)$.

We will now prove the following: for all $i \in \{1, \dots, k\},$ assuming all $v \in V_{\inbis, i}$ satisfy $\mu_{\inbis, i}^\Q(v) \geq \before(n_i^\Q),$ we have $\mu_{\out,i}^\Q(v) \geq \after(n_i^\Q).$ In particular, this proves by induction that vertices $v \in V_{\out, k}$ satisfy $\mu_{\out, k}^\Q \geq \after(n_k^\Q)$. Let $i \in \{1, \dots, k\},$ and assume that all $v \in V_{\inbis, i}$ satisfy $\mu_{\inbis, i}^\Q(v) \geq \before(n_i^\Q).$

To prove that all vertices $v \in V_{\out, i}$ satisfy $\mu_{\out, i}^\Q \geq \after(n_i^\Q)$, it suffices to show that vertices $v \notin V_{\rec, i}$ either satisfy $\mu^{\P}(v) \geq \after(n^\P)$ or $\mu^{\Q}(v) \geq \after(n_i^\Q)$, and conclude by applying the induction hypothesis to the recursive call to \ES$(n_i, \mu_{\rec, i})$. This amounts to showing that $\mu_{\rec, i}^{-1}(\Bset(n_i))=\varnothing.$ Assume for contradiction that $\mu_{\rec, i}^{-1}(\Bset(n_i)) \neq \varnothing$.

Since we have exited the while-loop on line~\ref{lin:while}, it holds that all vertices $v \in \mu_{\rec, i}^{-1}(\Bset(n_i))$ satisfy $\dest{\mu_{\rec, i}}{v} \in \Bset(n_i)$, that is, $\dest{\mu_{\rec, i}^\P}{v} < \after(n^\P)$ and $v$ is valid in $\mu_{\rec,i}^\Q$. This implies that $v$ belongs to the $\Q$-attractor of $(\mu_{\rec,i}^\Q)^{-1}(< \before(n_i^\Q))$ through $(\mu_{\rec,i}^\Q)^{-1}(\leq \before(n_i^\Q))$. In particular, there must be a vertex $v \in \mu_{\rec, i}^{-1}(\Bset(n_i))$ such that $\Q$ has a one-step strategy to reach the set $(\mu_{\rec,i}^\Q)^{-1}(\leq \before(n_i^\Q)) \setminus \mu_{\rec,i}^{-1}(\Bset(n_i))$. Now a vertex $u$ in this set must satisfy $\mu_{\rec,i}^\P(u) \geq \after(n^\P)$, hence Proposition~\ref{prop:one-step-strat} yields $\dest{\mu_{\rec, i}^\P}{v} \geq \after(n^\P),$ a contradiction.

There remains to prove that the while-loop in line~\ref{lin:sndwhile} empties $\Aset(n_k)$, that is, $\mu_\out^{-1}(\scope(n))=\mu_\out^{-1}(\Aset(n_k)) = \varnothing,$ which is achieved using a similar argument. Assume for contradiction that $\mu_\out^{-1}(\Aset(n_k)) \neq \varnothing$. Vertices $v \in \mu_\out^{-1}(\Aset(n_k))$ satisfy $\dest{\mu_\out^\P}{v} < \after(n^\P)$ and are valid in $\mu_\out^\Q$. This implies that they lie in the $\Q$-attractor to $(\mu_\out^\Q)^{-1}(< \after(n_k^\Q))$ through $(\mu_\out^\Q)^{-1}(\leq \after(n_k^\Q))$. Hence, there is $v \in \mu_\out^{-1}(\Aset(n_k)$ such that $\Q$ has a one-step strategy to $(\mu_\out^\Q)^{-1}(\leq \after(n_k^\Q)) \setminus \mu_\out^{-1}(\Aset(n_k))$. Then applying Proposition~\ref{prop:one-step-strat} yields a contradiction.
\end{proof}

\section{Complexity upper bound for Algorithm~\ref{alg:sym-main}} \label{app:complexity}

Let $\delta$ be the maximal degree of a node in $\Tc$ (that is, the maximal degree of a node in $\Tc^\Even$ or in $\Tc^\Odd$) and $N$ be the number of vertices in the game $\Gc$. We will show that Algorithm~\ref{alg:sym-main} performs $O(d \delta \min(|\Tc^\Even|, |\Tc^\Odd|))$. Note that in most instances\footnote{Here, the reader should understand ``In most known constructions of universal trees''.}, it holds that $\delta = O(N)$. If however one would want to implement the algorithm on a tree with high degree, we claim that by maintaining the support of the maps $\mu_\Odd$ and $\mu_\Even$, one may easily skip over subtrees with no vertices mapped to them when accelerating, and recover the $O(d N \min(|\Tc^\Even|, |\Tc^\Odd|))$ upper bound.

Consider an execution of Algorithm~\ref{alg:sym-main}. Given a node $n=(n^\P, n^\Q)$, we use $\mu_n$ to denote the value of the labellings just before a call to \ES~at node $n$, if there is one. We will also use $S_n$ to denote $\mu_n^{-1}(\scope(n))$. For $n \in \Tc$ such that a recursive call is made to \ES~ at $n$, we say that the call is \emph{accelerating} if it holds that for some $\R$, all vertices of $S_n$ are valid in $\mu_n^\R$. Note that in this case an acceleration occurs, hence no call to \ES~is made on any descendant of $n$ in $\Tc$. Otherwise, we say that the call is non-accelerating, in which case a recursive call is performed at each child of $n$.

Let us make two observations.

\medskip  \noindent \textbf{Observation 1.} 

There can be atmost $d\delta$ consecutive accelerating calls.

\medskip 
Indeed, if an accelerating call is made at a node $n$, the very next call is made at the smallest node which is not a descendant of $n$. It is easy to see to bound the size of such a chain of nodes by $d\delta$, $d$ being the total height of $\Tc$.

\medskip  \noindent \textbf{Observation 2.} 

Let $n=(n^\P, n^\Q) \in \Tc$ be a non-accelerating call. Then one of the following holds
\begin{itemize}
\item at least one lift on each labelling is performed in one of the following $\delta$ calls, or
\item one of the following $\delta$ calls is non-accelerating and is at has level $\level(n) -1$.
\end{itemize}

\medskip 

Indeed, if one of the children call at $n_1, \dots, n_k$ is non-accelerating, then the first non-accelerating call $n_i$ satisfies $n_{(i)}=n_i$ and $\level(n_i) = \level(n)-1$. Otherwise, all children call are accelerating, and an update for each player is performed.

\medskip 

A consequence of the second observation is that if $n$ is a non-accelerating call, there is a least one lift in each labelling in the following $d \delta $ calls. Hence, using the first observation, we conclude that every sequence of $2d \delta$ calls has at least one lift in each labelling, which implies the wanted upper bound.

\section{Equivalence of Algorithms~\ref{alg:uad} and Algorithm~\ref{alg:sym-main} equiped with short lifts and resets} \label{app:equivalence}

\subsection{Illustration and pseudo-code}

We first illustrate the variant of Algorithm~\ref{alg:sym-main} with short lifts and resets, then provide a full pseudo-code in Algorithm~\ref{alg:variant}. For completeness and convenience, we also recall the universal attractor decomposition algorithm from~\cite{JM20} (Algorithm~\ref{alg:uad}).

\begin{figure}[h]
  \makebox[\textwidth][c]{\scalebox{.8}{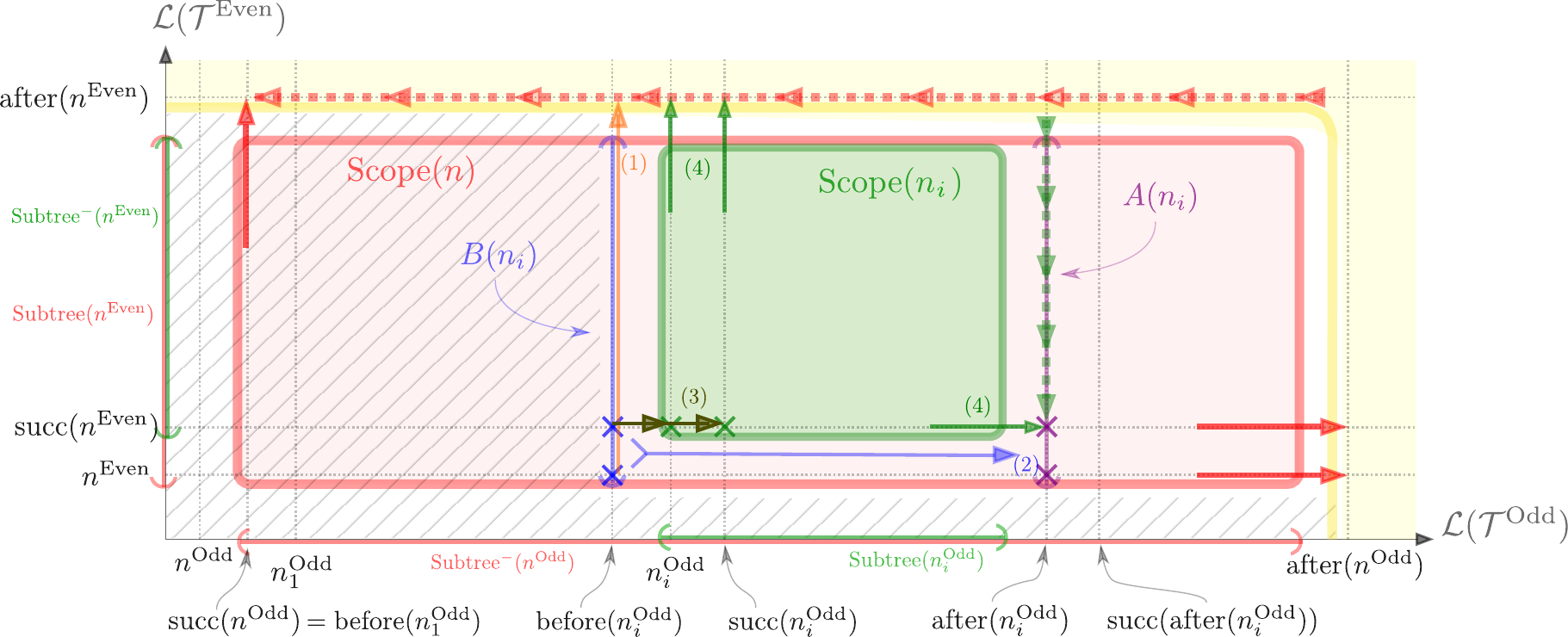}}
  
\caption{Illustration of the $i$-th iteration of the for-loop in a call to \ES~with short lifts and resets at node $n=(n^\P, n^\Q)$, in the case where $\P=\Even$. The reset in the recursive call at $n_i$ is symbolized by the green arrow, while the reset to the call at $n$ is symbolized by the red arrow. At the beginning of the $i$-th iteration of the for-loop, all vertices in $\scope(n)$ are mapped to the positions marked in blue. The orange (1), blue (2), and kaki (3) arrows respectively represent the actions of the while loops of lines~\ref{lin:while1},~\ref{lin:while2},~\ref{lin:while3}. At the beginning of the recursive call at $n_i$, all vertices in $\scope(n_i)$ are mapped to the positions marked in green. The green (4) arrow represents the action of the recursive call at $n_i$ (which includes the green reset). At the end of the $i$-th iteration, all vertices remaining in $\scope(n)$ are mapped to the positions marked in purple. We will argue that the orange (1) and the blue (2) arrows respectively correspond to $\Q$ and $\P$-attractor computations. After the global call at $n$, vertices that were initially in the scope are now mapped to the positions indicated by the red arrows.
}
\label{fig:alg-sym-with-resets}
\end{figure}

\begin{algorithm}[htb!]
\DontPrintSemicolon
  \SetKwFunction{solveE}{$\text{SolveEven}$}
  \SetKwFunction{solveO}{$\text{SolveOdd}$}
  \SetKwProg{fun}{procedure}{:}{}
  \fun{\solveE{$\G,d,n^\Even, n^\Odd$}}{
	 \If{$\G = \varnothing$}{ \Return{$\varnothing$} \;}
	 \Else{
	     $n_1^\Odd, \dots, n_k^\Odd \leftarrow$ children of $n^\Odd$ in $\Tc^\Odd$ \;
	     $\Gc_1 \leftarrow \Gc$ \;
		 \For{$i \leftarrow 1$ \KwTo $k$}{
  	     	$A_i \leftarrow \Attr{\Gc_i}{\Even}{\pi^{-1}(d) \cap \Gc_i}$ \;
			$\Gc'_i \leftarrow \G_i \setminus A$ \;
			$U_i \leftarrow$ \solveO{$\Gc'_i, d-1, n^\Even, n_i^\Odd$} \; \label{lin:recursive-call}
			$B_i \leftarrow \Attr{\Gc_i}{\Odd}{U_i}$ \;
			$\Gc_{i+1} \leftarrow \Gc_i \setminus B_i$ \;
  	 	}
  	 	\Return{$\Gc_{k+1}$} \;
	}
  }  
  \caption{A generic attractor-based algorithm.\label{alg:uad}}
\end{algorithm}

\begin{algorithm}[h]
\DontPrintSemicolon
  \SetKwFunction{emptyScope}{$\text{EmptyScope}$}
  \SetKwProg{fun}{procedure}{:}{}

  \fun{\emptyScope{$n, \mu$}}{
        $n \leftarrow (n^\P, n^\Q)$ \;
  		\If{all vertices in $\mu^{-1}(\scope(n))$ are valid in $\mu^\R$ for some $\R$\label{lin:ifacc}}{
			\ForAll{$v \in \mu^{-1}(\scope(n))$}{
				$\mu^{\bar \R}(v) \leftarrow \after(n^{\bar \R})$ \; 
			}
  		}
		\Else{\label{lin:else}
     		Let $n_1, \dots, n_k \leftarrow$ children of $n$ in $\Tc$ \;
  			\For{$i \leftarrow 1$ \KwTo $k$}{
  				\tcc{Three while loops resolve non-determinism}
				\While{there is $v \in \mu^{-1}(\Bset(n_i))$ such that $\dest{\mu^\P}{v} \geq \after(n^\P)$}{ \label{lin:while1}
					$\mu^P(v) \leftarrow \after(n^\P)$\;
				}
				\While{there is $v \in \mu^{-1}(\Bset(n_i))$ such that $\dest{\mu^\Q}{v} \geq \after(n_i^\Q)$}{ \label{lin:while2}
					$\mu^Q(v) \leftarrow \after(n_i^\Q)$\;
				}
				\While{there is $v \in \mu^{-1}(\Bset(n_i))$}{ \label{lin:while3}
					\If{$v$ has priority $\level(n_i^\Q)$}{
						$\mu^\Q(v) \leftarrow n_i^\Q$	\label{lin:lift1}\;	
					}
					\Else{
						$\mu^\Q(v) \leftarrow \succ(n_i^\Q)$	\label{lin:lift2}\;				
					}
				}
				\emptyScope$(n_i, \mu)$ \label{lin:recurse}\;
				}
			\tcc{Two while loops resolve non-determinism}
			\While{there is $v \in \mu^{-1}(\Aset(n_k))$ such that $\dest{\mu^\P}{v} \geq \after(n^\P)$}{\label{lin:sndwhile1}
				$\mu^\P(v) \leftarrow \after(n^\P)$ \;
			}
			\While{there is $v \in \mu^{-1}(\Aset(n_k))$}{\label{lin:sndwhile2}
				$\mu^\Q(v) \leftarrow \after(n^\Q)$ \label{lin:lift3}\;
			}
			
  		}
  		\tcc{We now perform a reset}
  		\ForAll{$v \in \mu^{-1}(\{\after(n^\P)\} \times_P \subtreem(n^\Q))$}{
  			$\mu^\Q(v) \leftarrow \succ(n^\Q)$ \;
  		}
  			
  	}
  
  \tcc{Main procedure:}
  Let $\mu^\Even,\mu^\Odd \leftarrow $ smallest labellings in $\Tc^\Even, \Tc^\Odd$ respectively \;
  Let $n \leftarrow $ root of $\Tc$ \;
  \emptyScope{$n, (\mu^\Even, \mu^\Odd)$}\;
  \Return{$(\mu^\Even, \mu^\Odd)$} \;
  \caption{The deterministic variant of Algorithm~\ref{alg:sym-main} with short lifts and resets, which is tailored to simulate the generic attractor-based algorithm.\label{alg:variant}}
\end{algorithm}

\newpage

\subsection{Formal proof of equivalence}

As we shall see, there is a case where Algorithm~\ref{alg:variant} is slightly more efficient, which we describe now in the context of \texttt{SolveEven} (there is a similar scenario for \texttt{SolveOdd}). If all vertices of $\Gc$ lie in the $A_1$, the Even attractor to vertices of (even) priority $d$, then it is easy to see that for all $i$, $U_i=\varnothing$, and $\G_i=\G$. Hence, \texttt{SolveEven} returns with $\Gc$ and performs $k$ (empty) recursive calls. However, in this precise case, Algorithm~\ref{alg:uad} immediately undergoes an acceleration, and terminates without performing any recursive calls. This explains the small (artificial) discrepancy $t \leq t' \leq \delta t$ in the statement of the Theorem, which could be improved with a more precise statement accounting for empty recursive calls.

In order to prove Theorem~\ref{thm:eq}, we now introduce a more precise inductive statement.

\begin{lemma}\label{lem:inductive-eq}
Let $\Gc' \subseteq \Gc$ be a subgame with priorities in $\{1, \dots, h\}$, and $n = (n^\P, n^\Q) \in \Tc$ with $\level(n)=h$. Let $\mu_\inbis=(\mu_\inbis^\Even, \mu_\inbis^\Odd)$ be a pair of labellings such that
\begin{itemize}
\item all $v \in \Gc \setminus \Gc'$ are mapped by $\mu_\inbis^\R$ to a position $\geq \after(n^\R)$ for some $\R \in \{\Even, \Odd\}$,
\item all $v \in \Gc$ satisfy $\mu_\inbis^\Q(v)=\succ(n^\Q)$, and
\item all $v \in \Gc$ of priority $h$ are mapped by $\mu_\inbis^\P$ to $n^\P$, and all $v \in \Gc$ of priority $< h$ to $\succ(n^\P)$.
\end{itemize}
Additionaly, we assume that for all $v \in \Gc$, $\dest{\mu_\inbis}{v} \in \scope(n)$, and that if $\Gc' \neq \varnothing$, there is a vertex in $\Gc'$ which is invalid in $\mu_\inbis^\Q$.

Let $\mu_\out$ be the labelling obtained after running \ES$(n,\mu_\inbis)$ \emph{with short lifts and resets, and $t$ be the number of recursive calls performed. Let $D$ be the output of }\texttt{SolveP}$(\Gc', h, n^\Even, n^\Odd)$ \emph{ and $t'$ be the number of recursive calls it performs.} Then
\begin{itemize}
\item $\mu_\out(D) \subseteq \{n^\P, \succ(n^\P)\} \times_\P \{\after(n^\Q)\}$,
\item $\mu_\out(\Gc' \setminus D) \subseteq \{\after(n^\P)\} \times_P \{\succ(n^\Q)\}$, and
\item $t \leq t' \leq (\delta +1) t$.
\end{itemize}
\end{lemma}

It is easy to see that Theorem~\ref{thm:eq} follows from Lemma~\ref{lem:inductive-eq} by applying it to the root of $\Tc$ and upon initialization of $\mu$.

\begin{proof}
This is proved by induction over the level $h$ of $n$. If $h=0$, necessarily $\Gc=\varnothing$, in which case $t=t'=1$, and $\mu_\out = \mu_\inbis$ satisfies the (vacuous) conclusions of the Lemma. We now assume that $h>1$ and that the result is known for nodes of smaller level.

We first rule out possible cases of acceleration in the call to \ES. If $\Gc' = \varnothing$, we have seen that there is nothing to prove. Assume $\Gc' \neq \varnothing$ and $t=1$, that is, the call to \ES~is accelerating. Since we have assumed that there is an invalid vertex in $\mu_\inbis^\Q$, it must be that all vertices of $\Gc'$ are valid in $\mu_\inbis^\P$. Let us prove that in this case, we have $\Attr{\Gc'}{\P}{\pi^{-1}(h) \cap \G'} = \G'.$ 

Assume for contradiction that $\Gc' \setminus \Attr{\Gc'}{\P}{\pi^{-1}(h) \cap \G'} \neq \varnothing.$ Note that vertices in this set must have priority $< h$, and hence are mapped to $\succ(n^\P)$ by $\mu_\inbis^\P$. Hence validity in $\mu_\inbis^\P$ of vertices in this set implies that they belong to the $\P$-attractor to $(\mu_\inbis^\P)^{-1}(\leq n^\P)$ through $(\mu_\inbis^\P)^{-1}(\leq \succ(n^\P)).$ Since these cannot be attracted by $\P$ to $\pi^{-1}(h) \cap \G' = \mu_\inbis^{-1}(\{n^\P\} \times_\P \{\succ(n^\Q)\})$, it must be that there is $v \in \Gc' \setminus \Attr{\Gc'}{\P}{\pi^{-1}(h) \cap \G'}$ from which $\Q$ has a one-step strategy to reach $(\mu_\inbis^\P)^{-1}(\leq \succ(n^\P)) \setminus \Gc'$. But such vertices are mapped by $\mu^\Q$ to $\geq \after(n^\Q)$, so Proposition~\ref{prop:one-step-strat} yields a contradiction to $\dest{\mu_\inbis}{v} \in \scope(n)$, which concludes that $\Attr{\Gc'}{\P}{\pi^{-1}(h) \cap \G'} = \G'$. In particular, in the call to \texttt{SolveP}, an easy induction implies that for all $i \in \{1, \dots, k\}$, $A_i=\Gc'$ and $U_i=\varnothing$, so the output $D$ is $\Gc_{k+1}=\Gc'$. Hence in this case $\mu_\out$ satisfies the conclusions of the Lemma, and we have $t=1$ and $t'=k+1 \leq \delta+1$.

Hence we now assume that the call to \ES~is non-accelerating. We need to introduce more notation. We let\footnote{Note that there is no notation clash as $\Gc'$ does not appear in the pseudo-code for Algorithm~\ref{alg:uad}.} $A_i,\Gc_i,\Gc'_i,U_i$ and $B_i$ for $i \in \{1, \dots, k\}$ as well as $D=\Gc_{k+1}$ be the subsets computed by the call to \texttt{SolveP}, in the provided pseudo-code. For all $i \in \{1, \dots, k\}$, we let $\mu_{i,0}, \mu_{i,1}, \mu_{i,2}$ and $\mu_{i,3}$ be the pair of labellings obtained respectively at the beginning of lines~\ref{lin:while1},~\ref{lin:while2},~\ref{lin:while3}, and~\ref{lin:recurse} in pseudo-code for Algorithm~\ref{alg:variant}, and\footnote{This defines $\mu_{k+1,0}$.} $\mu_{i,4}=\mu_{i+1,0}$ be the pairs of labellings obtained after the recursive call at line~\ref{lin:recurse}. Finally, we let $\mu_{k+1,1}$ and $\mu_{k+1,2}$ refer respectively to the labellings obtained after exiting the while-loops of lines~\ref{lin:sndwhile1} and~\ref{lin:sndwhile2}. 

We now prove by induction on $i \in \{1, \dots, k\}$ that
\begin{itemize}
\item $\Gc_i = \mu_{i,1}^{-1}(\scope(n)) = \mu_{i,1}^{-1}(\{n^\P, \succ(n^P)\} \times_P \{\before(n_i^\Q)\})$,
\item $A_i = \mu_{i,2}^{-1}(\Aset(n_i))$,
\item $\Gc'_i = \mu_{i,3}^{-1}(\scope(n_i))$ and $\mu_{i,3}$ satisfies the conditions for applying Lemma~\ref{lem:inductive-eq} at $n_i$ by induction,
\item $U_i = \Gc'_i \cap (\mu_{i,4}^{\P})^{-1}(\after(n^\P))$, and
\item $B_i = (\mu_{i+1,1}^{\P})^{-1}(\after(n^\P)) \cap \mu_{i,1}^{-1}(\scope(n))$.
\end{itemize}
We refer to these items as $(1^i),(2^i),(3^i),(4^i)$ and $(5^i)$ respectively. For the base case $i=1$, we only prove $(1^1)$, and argue that $(2^1),(3^1),(4^1)$ and $(5^1)$ are proved just as in the general case below. Note that the hypotheses of the Lemma over $\mu_\inbis$ imply that
\[
\mu_\inbis^{-1}(\Gc') = \mu_\inbis^{-1}(\scope(n)) = \mu_\inbis^{-1}(\{n^\P, \succ(n^\P)\} \times_\P \{\before(n_1^\Q)\}),
\]
since $\succ(n^\Q)=\before(n_1^\Q)$. Moreover, since vertices $v \in \Gc'$ satisfy $\dest{\mu_\inbis}{v} \in \scope(n)$, we have $\mu_\inbis = \mu_{1,1}$, so $(1^1)$ indeed holds.

We now turn to the inductive case: let $i \geq 2$ and assume all statements of the form $(l^j)$ for $l \in \{1, \dots, 5\}$ and $j \in \{1, \dots, i-1\}$ hold.
\begin{itemize}
\item \emph{Proof of $(1^i)$.} By $(5^{i-1})$, we have $B_{i-1} = (\mu_{i,1}^\P)^{-1}(\after(n^\P)) \cap \mu_{i-1,1}^{-1}(\scope(n))$, so $\G_i=\G_{i-1} \setminus B_{i-1} = \mu_{i-1,1}(\scope(n)) \setminus B_{i-1}$ (using $(1^{i-1})$) is precisely $\mu_{i,1}^{-1}(\scope(n))$ since in the for-loop, vertices can only be drawn out of $\scope(n)$ by being set to $\after(n^\P)$ in $\mu^\P$. The second equality in $(1^i)$ is given by the conclusion of the Lemma when it is applied inductively in $(3^{i-1})$.
\item \emph{Proof of $(2^i)$.} Note that $A_i=\Attr{\G_i}{\P}{\mu_{i,1}^{-1}(\{n^\P, \succ(n^\P)\} \times_\P \{\before(n_i^\Q)\})}$, and $\mu_{i,2}^{-1}(\Aset(n_i))$ is the set of vertices that have undergone an update in the while loop of line~\ref{lin:while2}. Assume for contradiction that $A_i \setminus \mu_{i,2}^{-1}(\Aset(n_i)) \neq \varnothing,$ and note that vertices in this set have destination $< \after(n_i^\Q)$ in $\mu_{i,2}^\Q$. We claim that either there is a vertex of priority $h$ in $A_i \setminus \mu_{i,2}^{-1}(\Aset(n_i))$, or there is a vertex in $A_i \setminus \mu_{i,2}^{-1}(\Aset(n_i))$ from which $\P$ has a one-step strategy to reach $\mu_{2,i}^{-1}(\Aset(n_i))$. In both cases, we reach a contradiction.

Conversely, assume $\mu_{i,2}^{-1}(\Aset(n_i)) \setminus A_i \neq \varnothing$. Let $v$ be the first vertex $\notin A_i$ to undergo a lift in the while-loop of line~\ref{lin:while2}, and let $\mu_{i,2.1}$ be the pair of labellings prior to this lift. By definition of the destination, it must be that $v$ is invalid in the $\Q$-labelling $\mu_{i,2.1}'^\Q$, which maps $v$ to the position $p^\Q \in \lazi{\Tc^\Q}$ preceding $\after(n_i^\Q)$ and is elsewhere identical to $\mu_{i,2.1}^\Q$. Since $v \in \Gc_i \setminus A_i,$ it must be that $\Q$ has a one-step strategy to reach $\Gc_i \setminus A_i$ from $v$ in $\Gc_i$. Note that these vertices are mapped to positions $< p^\Q$ by $\mu'$, except $v$ who is mapped to $p^\Q$. There are two cases.
\begin{itemize}
\item If $\Q$ has a one-step strategy to reach $(\mu_{i,2.1}'^{\Q})^{-1}{\leq p}$ in $\Gc$, then $v$ is valid in $\mu_{i,2.1}'^\Q$, a contradiction. This is clear if $p$ is a lazy position (which holds for $h \geq 3$), and also true in the case where $p = n_i^\Q$, which holds if $h \in \{1,2\}$.
\item Otherwise, $v$ must belong to $\P$, and have a successor $w$ with $\mu_{2.1}'^\Q(w) \geq \after(n_i^\Q)$. If $\mu_{2.1}'^\Q(w)=\after(n_i^\Q),$ then it must be that $w \in A_i$ so $v$ as well. Otherwise, we have $\mu_{2.1}'^\Q(w)=\mu_{\inbis}(w)=\after(n^\Q)$, and $\P$ has a one step strategy reach $w$ from $v$ in $\Gc$, so Proposition~\ref{prop:one-step-strat} gives a contradiction.
\end{itemize}

\item \emph{Proof of $(3^i)$.} We have $\Gc'_i= \Gc_i \setminus A_i$, and know that $A_i= \mu_{i,2}^{-1}(\Aset(n_i))=\mu_{i,3}^{-1}(\Aset(n_i))$, and that $\mu_{i,3}^{-1}(\Bset(n_i)) = \varnothing$. This implies that $\Gc'_i = \mu_{i,3}(\scope(n_i))$. The three items in the hypotheses of Lemma~\ref{lem:inductive-eq} are clearly satisfied. Moreover, if a vertex $v \in \Gc'_i$ were to satisfy $\dest{\mu_{i,3}}{v} \notin \scope(n_i)$, then it would have undergone an update in the previous while loops (of lines~\ref{lin:while1} or~\ref{lin:while2}). Lastly, assuming $\Gc'_i \neq \varnothing,$ there must be some invalid vertex in $\Gc'_i$ in $\mu_{i,3}^\Q$, otherwise all vertices in $\Gc'_i$ would already have been valid for $\Q$ at position $\before(n_i^\Q)$ in $\mu_{i,2}$, which is absurd.

\item \emph{Proof of $(4^i)$.} This is one of the conclusions of the inductive step in $(3^i)$.

\item \emph{Proof of $(5^i)$.} Note that $(\mu_{i+1,1}^{\P})^{-1}(\after(n^\P)) \cap \mu_{i,1}^{-1}(\scope(n)) \supseteq U_i$. It is not hard to reach a contradiction from $B_i \setminus (\mu_{i+1,1}^{\P})^{-1}(\after(n^\P)) \cap \mu_{i,1}^{-1}(\scope(n)) \neq \varnothing$ by applying Proposition~\ref{prop:one-step-strat}, so the left-to-right inclusion holds. The converse is analogous to the proof of $(2^i)$. 
\end{itemize}
Note that the proof of $(1^i)$ even carries on to $i=k+1$, showing that $\Gc_{k+1}=\mu_{k+1,1}^{-1}(\{n^\P, \succ(n^\P) \times_P \after(n_i^\Q)\})$. Hence, we obtain that $\mu_{k+2,2}$ maps vertices of $\Gc_{k+1}$ to $\{n^\P, \succ(n^\P) \times_P \after(n^\Q)\}$, and so does does $\mu_\out$ since these do not undergo a reset. Lastly, $\mu_{k+2,2}$ maps vertices of $\Gc' \setminus D$ to $\{\after(n^\P)\} \times_\P \subtreem(n^\Q)$, so after the reset, it indeed holds that $\mu_\out(\Gc' \setminus D) \subseteq \{\after(n^\P)\} \times_\P \{ \succ(n^\Q)\}$.
\end{proof}

%
%
%
%
%
%
\end{document}